\documentclass[reqno,11pt]{amsart}
\usepackage[utf8]{inputenc}
\usepackage{hyperref}
\usepackage{enumitem}
\usepackage{graphicx}
\usepackage{amscd, amsmath, amsthm}
\usepackage{slashed}
\usepackage{amssymb}
\usepackage{color}
\usepackage[mathscr]{eucal}
\numberwithin{equation}{section}
\allowdisplaybreaks[1]

\definecolor{ClaColor}{RGB}{255,0,0}
\newcommand{\clacomment}[1]{\textbf{\textcolor{ClaColor}{#1}}}

\definecolor{FeColor}{RGB}{0,0,255}

\newcommand{\SetFigFont}[3]{}

\title[The Fermionic Signature Operator in De Sitter Spacetime]{The Fermionic Signature Operator in \\
De Sitter Spacetime}

\author[C.\ Dappiaggi]{Claudio Dappiaggi}
\author[F.\ Finster]{Felix Finster}
\author[S.\ Murro]{Simone Murro}
\author[E.\ Radici]{Emanuela Radici \\ \\ February 2019}

\address{C. Dappiaggi, Dipartimento di Fisica \\ Universit{\`a} degli Studi di Pavia \& INFN, Sezione di Pavia, Via Bassi, 6 \\ I-27100 Pavia \\ Italy}
\email{claudio.dappiaggi@unipv.it}
\address{F. Finster, Fakult\"at f\"ur Mathematik \\ Universit\"at Regensburg \\ D-93040 Regensburg \\ Germany}
\email{finster@ur.de}
\address{S. Murro, Mathematisches Institut, Universit\"at Freiburg, 79104 Freiburg, Germany}
\email{simone.murro@math.uni-freiburg.de}
\address{E. Radici, DISIM - Department of Information Engineering, Computer
	Science and Mathematics, University of L'Aquila, Via Vetoio 1 (Coppito) 67100 L’Aquila (AQ) -
	Italy
}
\email{emanuela.radici@univaq.it}

\newtheorem{Def}{Definition}[section]
\newtheorem{Thm}[Def]{Theorem}
\newtheorem{Prp}[Def]{Proposition}
\newtheorem{Lemma}[Def]{Lemma}
\newtheorem{Remark}[Def]{Remark}
\newtheorem{Corollary}[Def]{Corollary}

\newcommand{\Thanks}{\vspace*{.5em} \noindent \thanks}
\newcommand{\beq}{\begin{equation}}
\newcommand{\eeq}{\end{equation}}
\newcommand{\Proof}{\begin{proof}}
\newcommand{\QED}{\end{proof} \noindent}
\newcommand{\QEDrem}{\ \hfill $\Diamond$}

\newcommand{\la}{\langle}
\newcommand{\ra}{\rangle}
\newcommand{\bra}{\mathopen{<}}
\newcommand{\ket}{\mathclose{>}}
\newcommand{\Sl}{\mathopen{\prec}}
\newcommand{\Sr}{\mathclose{\succ}}

\newcommand{\C}{\mathbb{C}}
\newcommand{\R}{\mathbb{R}}
\newcommand{\1}{\mbox{\rm 1 \hspace{-1.05 em} 1}}

\newcommand{\nuslsh}{\slashed{\nu}}
\renewcommand{\H}{\mathscr{H}}

\newcommand{\uslsh}{\slashed{u}}

\newcommand{\bep}{\begin{pmatrix}}
\newcommand{\enp}{\end{pmatrix}}

\newcommand{\Dir}{{\mathcal{D}}}
\newcommand{\D}{{\mathscr{D}}}

\renewcommand{\O}{{\mathscr{O}}}

\newcommand{\Lin}{\text{\rm{L}}}

\newcommand{\Cisc}{C^\infty_{\text{\rm{sc}}}}
\newcommand{\Cisco}{C^\infty_{\text{\rm{sc}},0}}

\DeclareMathOperator{\diag}{diag}

\newcommand{\p}{\mathfrak{p}}
\newcommand{\Sig}{\mathscr{S}}
\newcommand{\scrM}{\mycal M}
\newcommand{\scrN}{\mycal N}

\DeclareFontFamily{OT1}{rsfso}{}
\DeclareFontShape{OT1}{rsfso}{m}{n}{ <-7> rsfso5 <7-10> rsfso7 <10-> rsfso10}{}
\DeclareMathAlphabet{\mycal}{OT1}{rsfso}{m}{n}

\begin{document}

\begin{abstract}
The fermionic projector state is a distinguished quasi-free state for the algebra of Dirac fields in a globally hyperbolic spacetime. We construct and analyze it in the four-dimensional de Sitter spacetime, both in the closed and in the flat slicing. In the latter case we show that the mass oscillation properties do not hold due to boundary effects. This is taken into account in a so-called mass decomposition. The involved fermionic signature operator defines a fermionic projector state. In the case of a closed slicing,  we construct the fermionic signature operator and show that the ensuing state is maximally symmetric and of Hadamard form, thus coinciding with the counterpart for spinors of the Bunch-Davies state.
\end{abstract}

\maketitle
\tableofcontents

\section{Introduction}

Quantum fields in globally hyperbolic spacetimes and especially the algebraic formulation play a pivotal role in understanding and in formalizing different models, ranging from black hole physics to cosmology. Within the framework of algebraic quantum field theory, their analysis is based on a quantization scheme which can be summarized as a twofold approach, see {\it e.g.} \cite{Benini:2013fia,Brunetti:2015vmh}. On the one hand, one associates to a physical system a unital $*$-algebra $\mathcal{A}$, whose elements are interpreted as observables and which encodes structural properties such as causality and the canonical commutation or anti-commutation relations. On the other hand, one must select a state, that is a positive and normalized linear functional $\omega:\mathcal{A}\to\mathbb{C}$. Out of the pair $(\mathcal{A},\omega)$ one recovers the standard probabilistic interpretation of quantum theories via the GNS theorem. It associates to such pair a triple $(\mathcal{D}_\omega,\pi_\omega,\Omega_\omega)$, where $\mathcal{D}_\omega$ is a dense subspace of Hilbert space $\mathcal{H}_\omega$, $\pi_\omega:\mathcal{A}\to\mathcal{L}(\mathcal{D}_\omega)$ is a $*$-homomorphism, while $\Omega_\omega\in\mathcal{D}_\omega$ is a unit norm cyclic vector such that $\mathcal{H}_\omega=\overline{\pi_\omega[\mathcal{A}]}$.

Among the plethora of all possible states, one wants to select those which are physically sensible. Indeed, these are characterized mathematically by the so-called Hadamard condition, a constraint on the wavefront set of the underlying two point distribution, see \cite{Gerard:2019zzm,Khavkine:2014mta} for recent surveys on this topic.

While the existence of Hadamard states for free field theories on a generic globally hyperbolic spacetime $(\scrM,g)$ with $\dim\scrM>2$ has been established since long, the construction of explicit examples has been a subject of several investigations in the past years. Yet the attention has been focused mainly on bosonic scalar free field theories. On the contrary, if one considers spinor fields, which will be at the heart of this paper, only few  techniques have been thoroughly analysed, ranging from a positive frequency splitting on static backgrounds, to pseudodifferential calculus \cite{Gerard:2019zzm}, to more ad hoc methods {\it e.g.} \cite{Dappiaggi:2010gt,Hack:2015zwa,Hollands:1999fc} in which cosmological spacetimes were considered.

Another successful construction is known as the fermionic projector (FP) state. This is closely linked to the fermionic signature operator, which is a symmetric operator acting on the space of solutions of the massive Dirac equation on a globally hyperbolic spacetime. It was introduced in \cite{finite,infinite} and it has the advantage of producing a  distinguished quasi-free, pure state for the C$^*$-algebra of Dirac quantum fields, provided that a suitable condition, known as the strong mass oscillation property, holds true \cite{hadamard}. For a related analysis, containing a weaker
but non-canonical condition, refer to \cite{Drago:2016kue}. The advantage of focusing on the FP state is that it does not rely on the existence of any specific Killing isometry and thus it can be applied in a large class of scenarios. The price to pay for such generality is the impossibility to conclude a priori that the FP state is of Hadamard form. Several analyses of this issue have been made, and it is now clear that, although one cannot expect the Hadamard condition to hold true generically, see \cite{Fewster:2014wca}, it is nonetheless verified in many interesting scenarios, see in particular \cite{Finster:2016apv,hadamard,hadamard2}. We remark that the construction of the fermionic signature operator goes
back to~\cite{sea}, where the FP state was constructed perturbatively
in Minkowski space in the presence of an external potential.
More recently, a similar construction was proposed for scalar fields in~\cite{Afshordi:2012jf,johnston, sorkin2},
but only in space-times of finite lifetime. 
Yet, the ensuing state fails to obey to the Hadamard property as first observed in \cite{Fewster:2012ew}. This seems to be a generic feature which can be cured only by introducing suitable ad-hoc cut-off functions as proposed in \cite{Brum:2013bia}. This is in sharp contrast with the behaviour of the FP state in space-times of infinite lifetime.

In this paper we focus on a distinguished background, namely four-dimensional de Sitter spacetime, which is a maximally symmetric solution of the Einstein's equations with a positive cosmological constant. There exist two distinguished patches, which we consider. The first is the closed one which allows a global construction of the FP state on this background. The second is the flat one, which, despite covering only part of de Sitter spacetime, yields a metric which describes an exponentially expanding Universe with flat spatial sections. For this reason, we refer to this scenario as cosmological de Sitter. In both cases, the background is globally hyperbolic and one can therefore investigate whether the FP state exists and which are its properties. Besides the natural question whether the Hadamard condition holds true, in the closed patch, one can exploit the existence of a maximal number of Killing fields to construct a unique maximally symmetric two-point correlation function compatible with the Hadamard form, see \cite{Allen}. This identifies a unique, distinguished quasi-free state, to which we refer as the spinorial Bunch-Davies state, in analogy to the bosonic counterpart first constructed in \cite{Bunch:1978yq}.

Our investigation unveils that the two scenarios that we consider behave in a drastically different way. In the cosmological de Sitter background the strong mass oscillation property does not hold true and we can derive a so-called mass decomposition. This failure can be ascribed to the occurrence of boundary terms which originate from this spacetime being actually an open subset of the global de Sitter solution of the Einstein's equations. Yet, it turns out that the mass decomposition suffices to build a fermionic projector state, but we can conclude neither that it is of Hadamard form nor that it is the restriction of the global Bunch Davies state to the flat cover of de Sitter spacetime. 

On the contrary, when we focus on the closed slicing, as already observed in \cite{infinite}, the strong mass oscillation property holds true. Therefore we can construct the FP state and, since the procedure is automatically invariant under the action of all background isometries, we obtain a pure, quasi-free state for the algebra of Dirac fields whose two-point function is automatically maximally symmetric and, moreover, we prove that it is of Hadamard form. Hence the FP state coincides with the spinorial Bunch-Davies counterpart, proving once again the robustness of this constructive scheme.

The paper is organized as follows: In Section \ref{Sec:Preliminaries}, we review the basics on the Dirac equation on globally hyperbolic spacetimes. Subsequently we review the strong mass oscillation property, we introduce the fermionic signature operator and we highlight the construction of the fermionic projector (FP) state. To conclude we introduce succinctly the notion of Hadamard states focusing on the work of \cite{Allen} on maximally symmetric backgrounds. In Section \ref{Section:cosmological_dS} we focus our attention on the cosmological de Sitter spacetime and we derive a mass decomposition for the Dirac equation in Theorem \ref{thmdecomp}. Subsequently we construct the associated FP state and we discuss its properties. In Section \ref{Sec:closed_dS} we turn our attention to the closed slicing of de Sitter constructing explicitly the fermionic signature operator and the associated FP state. We conclude by proving that it is maximally symmetric and of Hadamard form.

\section{Preliminaries}\label{Sec:Preliminaries}
\subsection{The Dirac Equation in Globally Hyperbolic Spacetimes}
As in~\cite{finite, infinite}, we let~$(\scrM, g)$ be a smooth, globally hyperbolic Lorentzian spin
manifold, though of fixed dimension $k=4$.
For the signature of the metric we use the convention~$(+ ,-, -, -)$.
We denote the corresponding spinor bundle by~$S\scrM$. Its fibres~$S_x\scrM$ are endowed
with an inner product~$\Sl .|. \Sr_x$ of signature~$(2,2)$, which we refer to as the spin scalar product; for details see~\cite{baum, lawson+michelsohn}).
Clifford multiplication is described by a mapping~$\gamma$
which satisfies the anti-commutation relations,
\[ \gamma \::\: T\scrM \rightarrow \text{End}(S\scrM) \qquad
\text{with} \qquad \gamma(u) \,\gamma(v) + \gamma(v) \,\gamma(u) = 2 \, g(u,v)\,\1_{S(\scrM)} \:\,, \]
where~$\text{End}(S\scrM)$ is the endomorphism bundle.
We again write Clifford multiplication in components with the Dirac matrices~$\gamma^j$
and use the short notation with the Feynman dagger, $\gamma(u) \equiv u^j \gamma_j \equiv \uslsh$.
The metric connections on the tangent bundle and the spinor bundle are denoted by~$\nabla$.
The sections of the spinor bundle are also referred to as wave functions.
We denote with $C^\infty(\scrM, S\scrM)$ the smooth sections of the spinor bundle, while with $C^\infty_0(\scrM, S\scrM)$ those which are in addition compactly supported.
On the wave functions, one has the inner product
\begin{gather}
\bra .|. \ket \::\: C^\infty(\scrM, S\scrM) \times C^\infty_0(\scrM, S\scrM) \rightarrow \C \:, \notag \\
\bra \psi|\phi \ket = \int_\scrM \Sl \psi | \phi \Sr_x \: d\mu_\scrM,\:.  \label{stip} 
\end{gather}
where $d\mu_\scrM$ is the metric induced volume form on $\scrM$, while ~$\Sl .|. \Sr_x$ is the fiberwise inner product on $S_x\scrM$.
This inner product can be applied to more general wave functions, provided that
their pointwise inner product is integrable, i.e.~$\Sl \psi | \phi \Sr_x \in L^1(\scrM, d\mu_\scrM)$.
The Dirac operator~$\Dir$ is defined by
\[ \Dir := i \gamma^j \nabla_j \::\: C^\infty(\scrM, S\scrM) \rightarrow C^\infty(\scrM, S\scrM)\:. \]
For a given real parameter~$m \in \R$ (the ``mass''), the Dirac equation reads
\beq \label{Dirac}
(\Dir - m) \,\psi_m = 0 \:.
\eeq
The Cauchy problem for the Dirac equation is well-posed (see for example~\cite{taylor3, baer+ginoux, dimock3}).
For clarity, we always denote solutions of the Dirac equation by a subscript~$m$ and 
we consider mainly those lying in ~$\Cisc(\scrM, S\scrM)$, the space of smooth sections
with spatially compact support. Thereon, one defines the scalar product
\beq \label{print}
(\psi_m | \phi_m)_m = 2 \pi \int_\scrN \Sl \psi_m | \nuslsh \phi_m \Sr_x\: d\mu_\scrN(x) \:,
\eeq
where~$\scrN$ denotes any Cauchy surface and~$\nu$ its future-directed normal. Due to current conservation, 
the scalar product~$(.|.)_m$ is independent of the choice of~$\scrN$; for details see~\cite[Section~2]{finite}). Upon completion one obtains the Hilbert space~$(\H_m, (.|.)_m)$.

\subsection{The Strong Mass Oscillation Property}
In a spacetime of infinite life time, the inner product~$\bra \psi_m | \phi_m \ket$ between two solutions~$\psi_m, \phi_m \in \H_m$ is in general ill-defined, because the time integral in~\eqref{stip}
diverges. As observed in~\cite{infinite}, this problem can be avoided by
working with families of solutions of the Dirac operator and integrating over the mass parameter
before carrying out the spacetime integral in~\eqref{stip}.
We now summarize the parts of the construction needed for what follows.

We let~$I = (m_L, m_R) \subset \R^+$ be an open interval which does not contain zero.
We let~$\psi = (\psi_m)_{m \in I}$ be a family of solutions of the Dirac equation~\eqref{Dirac}.
We assume that for every~$m \in I$, the solution~$\psi_m$
is of the class~$\Cisc(\scrM, S\scrM)$. Moreover, we assume that the family
depends smoothly on~$m$ and vanishes for~$m$ outside a compact subset of~$I$.
We denote the corresponding class of functions by
\[ 
\psi \in \Cisco(\scrM \times I, S\scrM) \:, \]
where~$\Cisco(\scrM \times I, S\scrM)$ denotes the set of smooth functions which are compact both spatially and in~$I$.
Then for any fixed~$m$, we can take the scalar product~\eqref{print}. 
On families of solutions~$\psi, \phi \in \Cisco(\scrM \times I, S\scrM)$,
we introduce a scalar product by integrating over the mass parameter,
\[ 
( \psi | \phi) := \int_I (\psi_m | \phi_m)_m \: dm  \,,\]
where~$dm$ is the Lebesgue measure. Forming the completion gives the
Hilbert space $(\H, (.|.))$, whose norm is denoted by~$\| . \|$.

For the applications, it is useful to introduce
a subspace of solutions with useful properties:
\begin{Def} \label{defHinf}
We call~$\H^\infty \subset \Cisco(\scrM \times I, S\scrM) \cap \H$ a subspace
of the smooth solutions with the following properties:
\begin{itemize}[leftmargin=2.5em]
\item[\rm{(i)}] $\H^\infty$ is invariant under multiplication by smooth functions in the mass parameter,
\[ \eta(m)\, \psi(x,m) \in \H^\infty \qquad \forall\: \psi \in \H^\infty,\;\eta \in C^\infty(I) \:. \]
\item[{\rm{(ii)}}] The set~$\H^\infty_m := \{ \psi(.,m) \,|\, \psi \in \H^\infty\}$ is a dense subspace of~$\H_m$, i.e.
\[ \overline{\H^\infty_m}^{(.|.)_m} = \H_m \qquad \forall \:m \in I \:. \]
\end{itemize}
We refer to~$\H^\infty$ as the {\bf{domain}} for the mass oscillation property.
\end{Def} \noindent
In what follows, we always choose the maximal domain,
\[ \H^\infty = \Cisco(\scrM \times I, S\scrM) \cap \H \:. \]
We call ~$T$ the operator of multiplication by the mass parameter,
\[ T \::\: \H \rightarrow \H \:,\qquad (T \psi)_m = m \,\psi_m \:. \]
It is bounded, symmetric and leaves~$\H^\infty$ invariant,
\[ T^* = T \in \Lin(\H) \qquad \text{and} \qquad
T|_{\H^\infty} \::\: \H^\infty \rightarrow \H^\infty \:. \]
Moreover, integrating over~$m$ gives the operation
\[ \p \::\: \H^\infty \rightarrow \Cisc(\scrM, S\scrM)\:,\qquad
\p \psi = \int_I \psi_m\: dm \:. \]

\begin{Def} \label{defsMOP}
The Dirac operator~$\Dir$ on the globally hyperbolic manifold~$(\scrM, g)$ has the {\bf{strong mass oscillation
property}} in the interval~$I$ with domain~$\H^\infty$ (see Definition~\ref{defHinf}), if there exists a constant ~$c>0$ such that
\[ 
|\bra \p \psi | \p \phi \ket| \leq c \int_I \, \|\phi_m\|_m\, \|\psi_m\|_m\: dm
\qquad \forall\: \psi, \phi \in \H^\infty\:. \]
\end{Def} \noindent
For clarity, we point out that writing the inner product~$\bra \p \psi | \p \phi \ket$
implicitly involves the condition that this inner product must be well-defined and finite.
More precisely, one must verify that the function~$\Sl \p \psi | \p \phi \Sr_x \in L^1(\scrM, d\mu_\scrM)$
(for more details see~\cite[Sections~3 and~4]{infinite}).

\subsection{The Fermionic Signature Operator}
The construction of the fermionic signature operator is based on the
following results \cite[Theorem~4.2 and Proposition~4.4]{infinite}:

\begin{Thm} \label{thmsMOP}
The following statements are equivalent:
\begin{itemize}[leftmargin=2.5em]
\item[\rm{(i)}] The strong mass oscillation property holds.
\item[\rm{(ii)}] There exists a constant~$c>0$ such that for all~$\psi, \phi \in \H^\infty$,
the following two relations hold:
\begin{align}
|\bra \p \psi | \p \phi \ket| &\leq c\, \|\psi\|\, \|\phi\| \label{mb1} \\
\bra \p T \psi | \p \phi \ket &= \bra \p \psi | \p T \phi \ket\:. \label{mb2}
\end{align}
\item[\rm{(iii)}] There exists a family of linear operators~$\Sig_m \in \Lin(\H_m)$ which are uniformly bounded,
\[ \sup_{m \in I} \|\Sig_m\| < \infty\:, \]
such that
\beq \label{Smdef}
\bra \p \psi | \p \phi \ket = \int_I (\psi_m \,|\, \Sig_m \,\phi_m)_m\: dm \qquad
\forall\: \psi, \phi \in \H^\infty\:.
\eeq
\end{itemize}
\end{Thm}

\begin{Prp} {\bf{(uniqueness of~$\Sig_m$)}} \label{prpunique}
The family~$(\Sig_m)_{m \in I}$ from Theorem~\ref{thmsMOP}
can be chosen so that, for all~$\psi, \phi \in \H^\infty$, the expectation
value~$ f_{\psi, \phi}(m) := (\psi_m | \Sig_m \phi_m)_m$ is continuous in~$m$,
\beq \label{flip}
f_{\psi, \phi} \in C^0_0(I) \:.
\eeq
The family~$(\Sig_m)_{m \in I}$ with the properties~\eqref{Smdef} and~\eqref{flip} is unique.
Moreover, choosing two intervals~$\check{I}$ and~$I$ with~$m \in \check{I} \subset I$
and~$0 \not \in \overline{I}$, 
and denoting all the objects constructed in~$\check{I}$ with an additional check,
we have
\[ 
\check\Sig_m = \Sig_m \:. \]
\end{Prp}

These results yield for any~$m \in I$ a unique bounded
symmetric operator~$\Sig_m$ on~$\H_m$, referred to as the
{\em{fermionic signature operator}}. The construction is covariant and does not depend on the choice of Cauchy surfaces and observers, but it is non-local in the sense that it involves the global geometry of spacetime.

\subsection{The Fermionic Projector State}\label{Sec:fermionic_projector_state}
The fermionic signature operator gives rise to a distinguished vacuum
state of the quasi-free Dirac field, as we now recall.
First, the fermionic projector~$P$ is defined for a fixed mass parameter~$m \in I$
by (see~\cite[Definition~3.7]{finite} and~\cite[Definition~4.5]{infinite})
\beq \label{Pdef}
P := -\chi_{(-\infty, 0)}(\Sig_m)\, k_m 
\::\: C^\infty_0(\scrM, S\scrM) \rightarrow \H_m \:,
\eeq
where~$\chi_{(-\infty, 0)}(\Sig_m)$ is the projection onto the negative spectral subspace of
the fermionic signature operator, and~$k_m$ is the causal fundamental solution.
In addition, $P$ can be represented as an integral operator with a distributional kernel, generalizing to the case in hand the renown kernel theorem for scalar distributions (see~\cite[Theorem~4.7]{infinite}):
\begin{Thm}\label{Thm:2-pt} Assume that the strong mass oscillation property holds. Then
there exists a unique distribution~${\mathcal{P}} \in \D'(\scrM \times \scrM)$ such that 
for all~$\phi, \psi \in C^\infty_0(\scrM, S\scrM)$,
\[ \bra \phi | P \psi \ket = {\mathcal{P}}(\phi \otimes \psi) \:. \]
\end{Thm}

\noindent For our purposes, it is important to analyze the interplay between the fermionic projector state and the spacetime symmetries. This problem has been investigated in detail in \cite{Finster:2017wco} and, in the following lemma, we analyze this issue for the case in hand.

\begin{Lemma}\label{Lem:isometry_invariance}
	Under the assumptions of Theorem \ref{Thm:2-pt}, let $\alpha:\scrM\to\scrM$ be any continuous isometry of $(\scrM,g)$. Then, for all $\phi,\psi\in C^\infty_0(\scrM,S\scrM)$, 
	$$(\alpha^*\mathcal{P})(\phi\otimes\psi)=\mathcal{P}(\alpha_*(\phi\otimes\psi))=\mathcal{P}(\phi\otimes\psi).$$
\end{Lemma}

\begin{proof}
	Since $(\scrM, g)$ is per assumption a four-dimensional globally hyperbolic spacetime, 
	the spinor bundle is trivial, {\it i.e.} $S\scrM\simeq \scrM\times\mathbb{C}^4$, see \cite{Isham:1978ec}. Hence the action of any continuous isometry $\alpha$ of $(\scrM,g)$ lifts uniquely to the spinor bundle. With a slight abuse of notation, we indicate the lift simply with the same symbol $\alpha$. Consider any pair $(\phi,\psi)\in C^\infty(\scrM,S\scrM)\times C^\infty_0(\scrM,S\scrM)$ and let $\alpha$ be any isometry such that $\alpha_*\phi(x):=\phi(\alpha^{-1}(x))$ as well as $\alpha_*\psi(x)=\psi(\alpha^{-1}(x))$. Since the inner product \eqref{stip} is built out of the metric induced volume measure and of the fiberwise inner product on $\mathbb{C}^4$, it follows by construction that $\langle\phi|\psi\rangle=\langle\alpha_*(\phi)|\alpha_*(\psi)\rangle$. As a by-product, if we consider $\phi,\psi\in\mathcal{H}^\infty$, since the operator $\p$ commutes with the action of the isometries, being an integration over the mass, it holds 
	$$\langle \p\phi|\p\psi\rangle=\langle\alpha_*(\p\phi)|\alpha_*(\p\psi)\rangle,$$
	provided that the pairing is well-defined. As a consequence of Theorem \ref{thmsMOP}, this entails that 
	$$\int\limits_I(\psi_m|\Sig_m\phi_m)_m dm=\int\limits_I(\alpha_*\psi_m|\Sig_m\alpha_*\phi_m)_m$$
	As a consequence, we have realized each isometry $\alpha$ as a unitary operator $U_\alpha:\mathcal{H}_m\to\mathcal{H}_m$ so that $U_\alpha\psi:=\alpha_*(\psi)$ and $U^*_\alpha\Sig_m U_\alpha=\Sig_m$. Using the standard properties of spectral calculus it holds also that $-\chi_{(-\infty, 0)}(\Sig_m)=-U^*_\alpha\chi_{(-\infty, 0)}(\Sig_m)U_\alpha$. To conclude the proof it suffices to recall that $k_m$ is the causal fundamental solution of the Dirac equation, which is invariant under the action of the background isometries, that is $U_\alpha k_m=k_m U_\alpha$. Gathering together all these data, it holds that, for all $\phi,\psi\in C^\infty_0(\scrM, S\scrM)$, $\langle\alpha_*\phi|P\alpha_*\psi\rangle=\langle\phi|U^*_\alpha PU_\alpha\psi\rangle=\langle\phi|P\psi\rangle$, giving the result.
\end{proof}

On account of Theorem \ref{Thm:2-pt} and of Lemma \ref{Lem:isometry_invariance}, we can construct a distinguished state which is invariant under the action of all background isometries, characterized by the fact that
its two-point distribution coincides with~${\mathcal{P}}$
(see~\cite[Theorem~1.4]{hadamard} and the constructions in~\cite[Section~6]{hadamard}):
\begin{Thm} \label{thmstate}
There is an algebra of smeared fields generated by the abstract symbols~$\Psi(h)$, $\Psi^*(f)$
together with a pure quasi-free state~$\omega$ with the following properties: \\[0.3em]
{\rm{(a)}} The canonical anti-commutation relations hold:
\beq \label{antismeared}
\{\Psi(h),\Psi^*(f)\} = \bra h^* \,|\, \tilde{k}_m\, f \ket \:,\qquad
\{\Psi(h),\Psi(h')\} = 0 = \{\Psi^*(f),\Psi^*(f')\} \:.
\eeq
{\rm{(b)}} The two-point function of the state is given by
\[ \omega \big( \Psi(h) \,\Psi^*(f) \big) = - {\mathcal{P}}(h \otimes f)\:. \]
\end{Thm}

\subsection{Hadamard States}\label{Sec:Hadamard_states}

Among the many possible states for the algebra of smeared fields, only a suitable subclass is considered to be physically sensible. It is characterized mathematically by the so-called Hadamard condition. This is a prescription on the form of the wavefront set for the distribution 
$\omega_2\in\mathcal{D}^\prime(\scrM\times\scrM)$
 associated to the two-point function of an underlying state $\omega$, see \clacomment{\cite{Radzikowski:1996pa,Radzikowski:1996ei}}. In comparison to Theorem \ref{Thm:2-pt}, we do not use the symbol $\mathcal{P}$ to stress that the content of this section does not restrict to those two-point functions built out of a fermionic projector. 

Hadamard states have been thoroughly investigated for all free fields, see \cite{Khavkine:2014mta} for a recent review and, in the special case of spinors, see \cite{Kratzert,SahlmannVerch}. Most notably several results have been proven both in cosmological spacetimes \cite{Dappiaggi:2010gt,Hollands:1999fc} and in the framework of fermionic projectors \cite{hadamard, hadamard2}.

It is worth observing that the Hadamard condition can be translated equivalently into a constraint on the form of the integral kernel associated to the two-point distribution for the algebra of smeared fields. More precisely, as proven by Radzikowski in \cite{Radzikowski:1996pa,Radzikowski:1996ei} and working for definiteness with $\dim\scrM=4$, $\omega_2$ is of Hadamard form if and only if, for every geodesically convex neighbourhood $\mathcal{O}\subset \scrM$ and for every $\phi,\phi^\prime\in C^\infty_0(\mathcal{O},S\mathcal{O})$ where $S\mathcal{O}:= SM|_{\mathcal{O}}$, it holds
\begin{gather}\label{eq:local_Hadamard}
\omega_2(\phi,\phi^\prime)=\lim\limits_{\epsilon\to 0^+}\int\limits_{\mathcal{O}\times\mathcal{O}}d\mu_{\mathcal{O}}(x)d\mu_{\mathcal{O}}(y)\Sl I_{\epsilon}(x,y)||\phi(x)\phi^\prime(y)\Sr_{(x,y)}\\
I_{\epsilon}(x,y)=\frac{1}{8\pi^2}\mathcal{D}_{m,x}\left(\frac{U(x,y)}{\sigma_\epsilon(x,y)}+V(x,y)\ln\frac{\sigma_\epsilon(x,y)}{\lambda^2}+W(x,y)\right),\notag
\end{gather}
where $\mathcal{D}_m=\mathcal{D}-m$, while the subscript $x$ entails that the operator is acting on this variable. In addition $d\mu_{\mathcal{O}}$ indicates the metric induced volume measure in $\mathcal{O}$, while $\Sl |\Sr_{(x,y)}$ indicates the natural extension to $S\scrM\boxtimes S\scrM$, the external tensor product of $S\scrM$ with itself. Equation \ref{eq:local_Hadamard} was first discussed in \cite{Koehler,Verch} as a generalization of the counterpart for bosonic scalar theories, see \cite{Kay}. Next, the functions $U,V,W$ are smooth bispinors, while $\sigma_\epsilon(x,y)=\sigma(x,y)+2i\epsilon(t(x)-t(y))+\epsilon^2$, where $\sigma$ stands for half the squared geodesic distance. Finally, $t:\scrM\to\mathbb{R}$ is any time function. Note that, since $\mathcal{P}(x,y)$ is a weak solution of the equations of motion in both entries, both $U$ and $V$ are completely determined by the underlying geometry and by the dynamics. Hence the freedom in choosing a state is tantamount to that of fixing $W$. Although this local characterization of the Hadamard condition is fully equivalent to the microlocal one, it is in general only useful for local computations. As a matter of facts, the selection of a global, physically acceptable state would require the study of the integral kernel of the two-point distribution in every geodesic neighborhood. 

A notable exception is represented by maximally symmetric spacetimes, including thus the four-dimensional de Sitter spacetime, where, on the contrary, the local representation is a very efficient tool to construct Hadamard states, especially if one is interested in preserving all background isometries. For spinor fields, this problem has been investigated thoroughly by Allen and L\"utken in \cite{Allen}. Summarizing succinctly their work, they considered a four dimensional homogeneous manifold $(\scrM,g)$ whose Riemann tensor reads in components 
$$R_{abcd}=-\mathcal{R}^{-2}(g_{ac}g_{bd}-g_{ad}g_{bc}),$$
where $\mathcal{R}$ is a non vanishing constant. Focusing on de Sitter spacetime, for which $\mathcal{R}^2<0$ so that the scalar curvature $R=-\frac{12}{\mathcal{R}^2}>0$, they rewrote each Dirac spinor as  
$$\psi_\alpha=\left[\begin{array}{c}
\phi_A\\
\overline{\chi}_{\dot{A}}
\end{array}
\right],$$
where the dotted indices refer to the standard Van de Waerden notation, while the subscript $\alpha$ refers to the components of the Dirac spinor with respect to the standard basis in $\mathbb{C}^4$. Taking into account this notation, in \cite{Allen}, the attention is focused on $\omega(\Psi(g)\overline{\Psi}(f))$ where $\overline{\Psi}\gamma_0=\Psi^*$. In terms of components, the integral kernel of the two-point function can be rewritten as 
\begin{equation}\label{eq:aux1}
\omega_\alpha^{\beta^\prime}(x,y)=\left[
\begin{array}{cc}
hD_A^{B^\prime} & -f D_A^{A^\prime}n_{A^\prime}^{\dot{B}^\prime}\\
f D_{A}^{B^\prime}n^A_{\dot{A}} & h \overline{D}_{\dot{A}}^{\dot{B}^\prime}
\end{array}
\right],
\end{equation}
where unprimed and primed indices refer to quantities evaluated at $x$ and $y$ respectively. For notational convenience, we omitted the explicit dependence on $(x,y)$ on the right hand side. In \eqref{eq:aux1}, $D_A^{A^\prime}(x,y)$ is the bispinor which parallel transports a two-spinor $\phi^A$ at $x$ to a two-spinor $\chi_{A^\prime}$ at $y$, {\it i.e.} $\chi^{A^\prime}=\phi^A D_A^{A^\prime}$, {\it cf.} \cite{Allen}. At the same time, $n_{A\dot{A}}(x,y)=\nabla_{A\dot{A}}\mu(x,y)$, where $\mu(x,y)$ denotes the geodesic distance between $x$ and $y$. The remaining unknowns~$f$ and~$h$ in \eqref{eq:aux1} are two scalar functions which depend only on the geodesic distance $\mu$. They can be computed by imposing the equation of motion and, requiring that the short distance behaviour of the two-point function is compatible with \eqref{eq:local_Hadamard}, a unique solution can be constructed. Using the abbreviation 
\begin{equation}\label{eq:Zmu}
Z(\mu)=\cos^2\left(\frac{\mu}{2\mathcal{R}}\right) \:,
\end{equation}
we have
\begin{subequations}
\begin{equation}\label{eq:sol1a}
f(\mu)=-\frac{i\Gamma(a)\Gamma(b)}{8\sqrt{2}\pi^2|\mathcal{R}|^3}\sqrt{1-Z(\mu)}\:F(a,b;2;Z(\mu)) \:,
\end{equation}
\begin{equation}\label{eq:sol1b}
h(\mu)=-\frac{m\Gamma(a)\Gamma(b)}{32\pi^2|\mathcal{R}|^2}\sqrt{Z(\mu)}\:F(a,b;3;Z(\mu)) \:,
\end{equation}
\end{subequations}
where $F$ stands for the Gaussian hypergeometric function, while 
$$a=2+\sqrt{m^2\mathcal{R}^2}\qquad\textrm{and}\qquad b=2-\sqrt{m^2\mathcal{R}^2}.$$
Equation \eqref{eq:aux1} identifies the unique, maximally symmetric, two-point function on a de Sitter spacetime which is compatible with the Hadamard condition. The unique quasi-free state, which is unambiguously constructed out of \eqref{eq:aux1}, will be referred to as the {\em spinorial Bunch-Davies state}.

\section{Cosmological De Sitter}\label{Section:cosmological_dS}
\subsection{The Dirac Equation and its Separation}
In this section we consider cosmological de Sitter spacetime, that is we work in the so-called flat slicing of de Sitter, see {\it e.g.} \cite{Moschella:2006pkh}. Without loss of generality we assume the spacetime to be 4 dimensional. In this chart one introduces the so-called {\em{cosmological time}} $t \in \R$ where the metric reads
\begin{equation}\label{flat_dS_metric}
ds^2=dt^2-R(t)^2 \sum_{\alpha=1}^{3} dx^2_\alpha \qquad \text{with} \qquad
R(t) :=e^{t} \:.
\end{equation}

Observe that, although in the flat slicing, the four dimensional cosmological de Sitter spacetime is diffeomorphic to $\mathbb{R}^4$, it represents only an open subset of the full de Sitter spacetime, which will be discussed in Section \ref{Sec:closed_dS}.

The Dirac operator in the flat slicing was computed in~\cite{moritz} to be
\beq
\Dir = i \gamma^0 \left( \partial_t + \frac {3 \dot{R}(t)}{2R(t)} \right) + \frac{1}{R(t)}
\begin{pmatrix} 0 & \Dir_{\R^{3}} \\ -\Dir_{\R^{3}} & 0 \\ \end{pmatrix} , \label{dirac_a}
\eeq
where~$\Dir_{\R^{3}}$ is the Dirac operator on~$\R^{3}$, i.e.
\[ \Dir_{\R^{3}} = \sum_{\alpha=1}^3 i \sigma^\alpha \partial_\alpha \,, \]
and $\sigma^\alpha$ are the three Pauli matrices.
The inner products \eqref{stip} and \eqref{print} take the form
\begin{align}
\bra \psi | \phi \ket &= \int_{-\infty}^\infty dt \int_{\R^3} \Sl \psi | \phi \Sr(t,x) \:R(t)^3 d^3x\:  \label{stip2h} \\
( \psi_m | \phi_m )_m &= 2 \pi \int_{\R^3} \Sl \psi | \gamma^0 \phi \Sr(t,x) \:R(t)^3\,  d^3x \:, \label{print2h}
\end{align}
where~$\Sl \psi | \phi \Sr = \psi^\dagger \gamma^0 \phi$ and~$\gamma^0=\diag(1,1,-1,-1)$.

The Dirac equation can be separated as follows.
First, given~$k \in \R^3$, the spatial Dirac operator can be diagonalized by the
separation ansatz
\beq \label{phikdef}
\phi_{k, \pm}(x) = e^{i k x}\: \chi_{k, \pm} \:,
\eeq
where the two spinors~$\chi_{k,\pm}$ form orthonormal eigenvector basis of the matrix~$k \sigma$, i.e.\
\beq \label{spinON}
k \sigma \chi_{k,s} = \pm |k|\, \chi_{k,s} \:,\qquad \la \chi_{k,s}, \chi_{k,s'} \ra_{\C^2} = \delta_{s,s'}
\eeq
for~$s,s' \in \{\pm\}$. A straightforward computation yields that these spinors
can be chosen explicitly as
\[ 
\begin{split}
\chi_{k,+}& = \frac{1}{\sqrt{2\, |k| \,\big( |k|-k_3 \big)}} \begin{pmatrix} -|k|+k_3 \\ k_1 + i k_2 \end{pmatrix} \\
\chi_{k,-} &= \frac{1}{\sqrt{2\, |k| \,\big(|k|+k_3 \big)}} \begin{pmatrix} |k|+k_3 \\ k_1 + i k_2 \end{pmatrix} \:.
\end{split} 
\]
Employing the ansatz
\begin{equation} \label{ansatzhalf}
\psi_m = R(t)^{-\frac{3}{2}}
\begin{pmatrix} u_1(m,t) \:\phi_{k, \pm}(x) \\
u_2(m,t) \:\phi_{k, \pm}(x) \end{pmatrix} \:,
\end{equation}
the Dirac equation for the Dirac operator \eqref{dirac_a} gives rise to the ODE 
\beq \label{DiracHalf}
i\, \frac{d}{dt} \begin{pmatrix} u_1 \\ u_2 \end{pmatrix}
=  \begin{pmatrix} m & -\lambda\, e^{-t} \\ -\lambda\, e^{-t} & - m \end{pmatrix}
\begin{pmatrix} u_1 \\ u_2 \end{pmatrix}  \qquad \text{with} \qquad \lambda := \pm |k| \:.
\eeq
Starting from ~$t \rightarrow \infty$, the asymptotic future, the exponential decay of the matrix elements suggests that the solutions should behave like plane waves~$\sim e^{\pm i m t}$. This can be made precise following \cite[Lemma~6.3]{infinite}, evaluating the error with a Gr\"onwall estimate:

\begin{Lemma} \label{lemmagronwall}
Asymptotically as $t \rightarrow +\infty$, every solution of~\eqref{DiracHalf} is of the form
\beq
u(t) = \begin{pmatrix} e^{-imt} \:f^\infty_1 \\[0.3em]
e^{imt} \:f^\infty_2 \end{pmatrix} + E(t) \label{eq:3za}
\eeq
with the error term bounded by
\beq
\|E(t)\| \leq \big\| f^\infty \big\| \,\exp \big( |\lambda| \,e^{-t} \big) \:. \label{eq:3c}
\eeq
\end{Lemma}
\Proof Substituting into~\eqref{DiracHalf} the ansatz
\beq
u(t) = \begin{pmatrix} e^{-imt} \:f_1(t) \\
e^{imt} \:f_2(t) \end{pmatrix} \;, \label{eq:36}
\eeq
gives
\[
\frac{df}{dt} = -\lambda e^{-t}  \begin{pmatrix}
0 & e^{2imt} \\ e^{-2imt} & 0 \end{pmatrix}  f,\quad\: f(t):=\begin{pmatrix}
f_1(t)\\ f_2(t)
\end{pmatrix} \:. \]
Taking the norm, we obtain the differential inequality
\beq \label{eq:36a}
\left\| \frac{df}{dt} \right\| \leq |\lambda|\: e^{-t} \:\|f\| \:.
\eeq

Let us first show that~$f(t)$ has a limit as~$t \rightarrow \infty$.
To this end, we apply Kato's inequality to~\eqref{eq:36a},
\[ 
\frac{d}{dt}\, \|f\| \leq |\lambda|\, e^{-t}\:\|f\| \:. \]
We may assume that our solution is nontrivial, so that~$\|f\| \neq 0$.
Thus we may divide by~$\|f\|$,
\beq \label{logineq}
\frac{d}{dt} \log \|f\| \leq |\lambda|\: e^{-t} \:.
\eeq
Since the right hand side is integrable, we conclude that~$\log |f|$ has bounded variation,
implying that~$\log |f|$, and therefore also~$f$ converges as~$t\rightarrow \infty$.
We set~$f^\infty = \lim_{t \rightarrow \infty} f(t)$.

In order to estimate~$\|f-f^\infty\|$, we integrate~\eqref{logineq} from~$t$ to any~$t_{\max}>t$,
\[ \|f(t)\| \leq \|f(t_{\max})\| \exp \left( \int_{t}^{t_{\max}} |\lambda|\:e^{-\tau}\: d\tau \right) \:. \]
Substituting this inequality into~\eqref{eq:36a} yields
\begin{align*}
\left\| \frac{df}{dt} \right\| &\leq |\lambda|\:e^{-t}\: \|f(t_{\max})\| \exp \left( \int_{t}^{t_{\max}} |\lambda|
\: e^{-\tau}\: d\tau \right) \\
&= \|f(t_{\max})\| \:\frac{d}{dt} \exp \left( \int_{t}^{t_{\max}} |\lambda|\:e^{-\tau}\: d\tau \right) .
\end{align*}
Integrating on both sides from~$t$ to some~$t_{\max}$ gives
\[ \left\| f(t) - f(t_{\max}) \right\| \leq \|f(t_{\max})\| \:\exp \left( \int_{t}^{t_{\max}} |\lambda|\:e^{-\tau}\:
d\tau \right) . \]
Now we take the limit~$t_{\max} \rightarrow \infty$ to obtain
\[ \left\| f(t) - f^\infty \right\| \leq \|f^\infty\| \:\exp \left( \int_t^{\infty} |\lambda|\:e^{-\tau}\:
d\tau \right) = \|f^\infty\| \:\exp \big( |\lambda|\:e^{-t} \big) \:. \]
Using this estimate in~\eqref{eq:36} and comparing with~\eqref{eq:3za} gives the desired estimate for~$E$.
\QED

It is most convenient to denote a fundamental system of solutions according to its
asymptotics at large times. Thus for the separation constants
\begin{align*}
k \in \R^3 &\qquad \text{spatial momentum} \\
s \in \{\pm \} &\qquad \text{spin orientation} \\
a \in \{1,2\} &\qquad \text{frequency in asymptotic future}
\end{align*}
we introduce the Dirac solutions
\beq \label{fundamental}
\psi_m^{k,s,a}(t,x) = R^{-\frac{3}{2}}
\begin{pmatrix} u^{s,a}_1(m,k,t) \:\phi_{k, s}(x) \\
u^{s,a}_2(m,k,t) \:\phi_{k, s}(x) \end{pmatrix} ,
\eeq
where the vectors~$u^{s,a}$ have the asymptotics
\beq \label{uasy}
\lim_{t \rightarrow \infty} e^{imt}\, u^{s,a}_1 = \delta^a_1 \:,\qquad
\lim_{t \rightarrow \infty} e^{-imt}\, u^{s,a}_2 = \delta^a_2 \:.
\eeq

We need to study the asymptotic behaviour of the solutions of the Dirac equation as ~$t \rightarrow -\infty$, which corresponds to the boundary of cosmological de Sitter, if realized as an open subset of the whole de Sitter background. This problem is best tackled considering the 
\[ 
\text{\em{conformal time}} \qquad \tau := -e^{-t} \in \R^- \:, \]
so that $t\to -\infty$ corresponds to the limiting case~$\tau \rightarrow -\infty$. Likewise, taking the limit $t\to \infty$ translates to $\tau\to 0^-$.
The transformation
\[ \frac{d}{dt} = \frac{d\tau}{dt}\: \frac{d}{d\tau} = e^{-t}\: \frac{d}{d\tau} = -\tau\: \frac{d}{d\tau} \]
gives rise to the ODE
\begin{align}
-i\, \frac{d}{d\tau} \begin{pmatrix} u_1 \\ u_2 \end{pmatrix}
&=  \begin{pmatrix} m/\tau & \lambda \\ \lambda & -m/\tau \end{pmatrix}
\begin{pmatrix} u_1 \\ u_2 \end{pmatrix} . \label{Diractau}
\end{align}

\subsection{The Mass Decomposition}
As we shall see, due to boundary terms, the Dirac operator in cosmological de Sitter does {\em{not}} have
the strong mass oscillation property. Instead, we shall derive a so-called {\em{mass decomposition}}, see Theorem~\ref{thmdecomp}.

We first specify the domain~$\H^\infty$ in Definition \ref{defHinf}.
To this end, we consider a superposition of the fundamental solutions in~\eqref{fundamental},
\beq \label{superpose}
\psi_m(t,x) = \int_{\R^3} \frac{d^3k}{(2 \pi)^3} \sum_{s=\pm} \sum_{a=1,2} \;
\hat{\psi}(m,k,s,a)\: \psi_m^{k,s,a}(t,x) \:.
\eeq
For convenience, we choose the domain~$\H^\infty$ for the mass oscillation property
as superpositions which are smooth and compactly supported both in the mass and in the momentum variables. Moreover, in order to avoid technical problems
at~$k=0$, we assume that the wave functions vanish in a neighborhood of~$k=0$.
Thus we choose~$\H^\infty$ as the space of functions of the form~\eqref{superpose} with
\beq \label{Hinf}
\hat{\psi} \in C^\infty_0 \big( I \times (\R^3 \setminus \{0\} ) \times \{\pm 1\} \times \{1,2\} \big) \:.
\eeq

Applying Plancherel's theorem and using \eqref{spinON}, we write the scalar product~\eqref{print2h} as
\begin{align}
( \psi_m | \tilde{\psi}_m )_m &= 2 \pi \int_{\R^3} \frac{d^3k}{(2 \pi)^3}\:
\sum_{s=\pm} \;\sum_{a,a'=1,2} \overline{\hat{\psi}(m,k,s,a)} \: \hat{\tilde{\psi}}(m,k,s,a') \notag \\
&\qquad \times \la u^{s,a}(m,k,t), u^{s,a'}(m,k,t) \ra_{\C^2} \label{print3h} \:.
\end{align}
From this formula it descends that the solutions of the
form~\eqref{superpose} with~$\hat{\psi}$ of the form~\eqref{Hinf} are dense in~$\H$.
The inner product~\eqref{stip2h} can be written in cosmological and conformal time as
\begin{align}
\bra \psi_m | \tilde{\psi}_{m'} \ket &= \int_{-\infty}^\infty dt \int_{\R^3} \frac{d^3k}{(2 \pi)^3}\:
\sum_{s=\pm} \;\sum_{a,a'=1,2}
\overline{\hat{\psi}(m,k,s,a)} \: \hat{\tilde{\psi}}(m',k,s,a') \notag \\
&\qquad \times \Big\la u^{s,a}(m,k,t), \begin{pmatrix} 1 & 0 \\ 0 & -1 \end{pmatrix} u^{s,a'}(m',k,t) \Big\ra_{\C^2} \label{stip3h} \\
&= \int_{-\infty}^0 \frac{d\tau}{|\tau|} \int_{\R^3} \frac{d^3k}{(2 \pi)^3}\:
\sum_{s=\pm} \;\sum_{a,a'=1,2}
\overline{\hat{\psi}(m,k,s,a)} \: \hat{\tilde{\psi}}(m',k,s,a') \notag \\
&\qquad \times \Big\la u^{s,a}(m,k,\tau), \begin{pmatrix} 1 & 0 \\ 0 & -1 \end{pmatrix} u^{s,a'}(m',k,\tau) \Big\ra_{\C^2} \:.\label{stip3tau}
\end{align}
Here one should keep in mind that the $t$- respectively $\tau$-integrals
in general do not exist. Therefore, the formulae~\eqref{stip3h} and~\eqref{stip3tau} are
to be understood merely as formal expressions which still need to be given a mathematical meaning.

In the following we need to derive suitable decay estimates for the solutions of the Dirac equations both near the boundary and near infinity. In the latter case, these estimates can be
obtained similarly to ~\cite[Lemma~6.4]{infinite} by taking derivatives with respect to the mass and integrating by parts:
\begin{Lemma} \label{lemmainfty}
For any~$\psi \in \H^\infty$ there is a constant~$c>0$ such that
\[ \big\| (\p \psi)(t,.) \big\|_{L^2(\R^3)} \leq \frac{c}{t} \qquad \text{for all~$t > 1$}\:. \]
\end{Lemma}
\Proof Using the representation~\eqref{eq:3za} in~\eqref{superpose},
the contribution by the error term~$E(t)$ decays exponentially in time, uniformly
in~$\vec{k}$ and~$m$, i.e.
\begin{align*}
\bigg\| \begin{pmatrix} E^a_1(m,k,t) \:\phi_{k, s}(x) \\
E^a_2(m,k,t) \:\phi_{k, s}(x) \end{pmatrix} \bigg\| \leq c\: e^{-t}.
\end{align*}
Note that~$\lambda =\pm |k|$ is bounded. Moreover, the contribution by the plane waves in~\eqref{eq:3za} can be dealt with using integration by parts,
\begin{align*}
\int_I & \hat{\psi}(m,k,s,a)\: 
\begin{pmatrix} e^{-imt} \:f^\infty_1(m,k) \:\phi_{k, s}(x) \\
e^{imt} \:f^\infty_2(m,k) \:\phi_{k, s}(x) \end{pmatrix}\: dm \\
&=-\frac{1}{it} \int_I \hat{\psi}(m,k,s,a)\: 
\bigg( \frac{d}{dm} 
\begin{pmatrix} e^{-imt} & 0 \\
0 & -e^{imt} \end{pmatrix} \bigg) 
\begin{pmatrix} f^\infty_1(m,k) \:\phi_{k, s}(x) \\
f^\infty_2(m,k) \:\phi_{k, s}(x) \end{pmatrix}\: dm \\
&= \frac{1}{it} \int_I \partial_m \hat{\psi}(m,k,s,a)\: 
\begin{pmatrix} e^{-imt} \:f^\infty_1(m,k) \:\phi_{k, s}(x) \\
-e^{imt} \:f^\infty_2(m,k) \:\phi_{k, s}(x) \end{pmatrix}\: dm \\
&\quad +\frac{1}{it} \int_I \hat{\psi}(m,k,s,a)\: 
\begin{pmatrix} e^{-imt} \:\partial_m f^\infty_1(m,k) \:\phi_{k, s}(x) \\
-e^{imt} \:\partial_m f^\infty_2(m,k) \:\phi_{k, s}(x) \end{pmatrix}\: dm \:.
\end{align*}

The obtained expressions are smooth and compactly supported in~$k$.
Therefore, their $L^2$-norm in~$k$ is bounded. Applying Plancherel's theorem the sought result  descends.
\QED

As next step, we study the asymptotics of the ODE~\eqref{Diractau} as~$\tau \rightarrow -\infty$, corresponding to the boundary of the cosmological de Sitter spacetime.
In this case, the matrix on the right of~\eqref{Diractau} tends to a constant matrix, giving rise to oscillations proportional to $e^{\pm i \lambda \tau}$. However, the matrix entries in \eqref{Diractau} which decay like~$1/\tau$ are not integrable in~$\tau$, making it impossible to apply Gr\"onwall estimates.
In order to bypass this problem, our method is to diagonalize the matrix in~\eqref{Diractau} with a unitary matrix~$U(\tau)$
for every~$\tau$. Clearly, the $\tau$-derivative of~$U$ gives rise to
an error term, but this term decays quadratically in~$\tau$, making it possible to
use Gr\"onwall estimates. This method was used previously in~\cite[Lemma~3.5]{tkerr}.
\begin{Lemma} \label{lemmaboundary}
If~$\lambda \neq 0$, as~$\tau \rightarrow -\infty$
the solutions of \eqref{Diractau} have the asymptotics
\beq \label{uasybd}
u(\tau) = U(\tau) \begin{pmatrix} g_1\: e^{i \varphi(\tau)} \\ g_2\: e^{-i \varphi(\tau)} \end{pmatrix} 
+ E(\tau) \:,
\eeq
with~$g_1, g_2 \in \C$, while the error term~$E(\tau)$ is bounded by
\beq \label{Etau}
\big| E(\tau) \big| \leq \frac{c}{|\tau|} \qquad \text{for all~$\tau<-1$}\:.
\eeq
Here the functions~$\varphi(\tau)$ and~$U(\tau)$ are given by
\begin{align}
\varphi(\tau) &= \int^\tau \sqrt{\lambda^2 + \frac{m^2}{\tilde{\tau}^2}}\: d\tilde{\tau} \label{phidef} \\
U(\tau) &= \begin{pmatrix} \cos \alpha & \sin \alpha \\
-\sin \alpha & \cos \alpha \end{pmatrix} \qquad \text{with} \qquad
\alpha = -\frac{1}{2}\: \arctan \Big( \frac{\lambda \tau}{m} \Big) \:. \label{Udef}
\end{align}
\end{Lemma}
\Proof We write~\eqref{Diractau} as
\[ -i \:\frac{du}{d\tau} = A u \qquad \text{with} \qquad A :=
\begin{pmatrix} m/\tau & \lambda \\ \lambda & -m/\tau \end{pmatrix} \:. \]
$A$ is diagonalized by~$U$, i.e.\
\[ U^{-1} A U = \begin{pmatrix} \sqrt{\lambda^2 + m^2/\tau} & 0 \\ 0 & -\sqrt{\lambda^2 + m^2/\tau} \end{pmatrix} \:. \]
As a consequence,
\[ -i \:\frac{d}{d\tau} \big( U^{-1} u \big) = 
\begin{pmatrix} \sqrt{\lambda^2 + m^2/\tau} & 0 \\ 0 & -\sqrt{\lambda^2 + m^2/\tau} \end{pmatrix}\: \big( U^{-1} u \big)
+i \,\big(U^{-1} \:\partial_\tau U \big) \big( U^{-1} u \big) \:. \]
Hence the vector~$g(\tau)$ defined by
\[ g(\tau) = \begin{pmatrix} e^{i \varphi(\tau)} & 0 \\ 0 & e^{-i\varphi(\tau)} \end{pmatrix} U^{-1}(\tau)\, u(\tau) \]
satisfies the ODE
\[ -i \:\frac{dg}{d\tau} = i \,
\begin{pmatrix} e^{i \varphi(\tau)} & 0 \\ 0 & e^{-i\varphi(\tau)} \end{pmatrix} \big(U^{-1} \partial_\tau U \big) 
\begin{pmatrix} e^{-i \varphi(\tau)} & 0 \\ 0 & e^{i\varphi(\tau)} \end{pmatrix} g \:. \]
Taking the norm, we obtain
\[ \Big\| \frac{dg}{d\tau} \Big\| \leq \big\| U^{-1} \partial_\tau U \big\| \: \|g\|
= \frac{1}{2 \tau^2}\: \frac{m \lambda}{\lambda^2 + m^2/\tau^2}\:
\Big\| \begin{pmatrix} 0 & -1 \\ 1 & 0 \end{pmatrix} \Big\| \: \|g\|\:. \]
For any~$\lambda \neq 0$, the right hand side decays quadratically.
Similar to the proof of Lemma~\ref{lemmagronwall}, we can again apply a Gr\"onwall
estimate to obtain the sought result.
\QED
Equation \eqref{phidef} determines the phase~$\varphi$ only up to an integration constant.
For the following estimates, it is most convenient to fix this constant by choosing
\beq \label{phichoice}
\varphi(\tau) := |\lambda|\, \tau + \int_{-\infty}^\tau\: \bigg( \sqrt{\lambda^2 + \frac{m^2}{\tilde{\tau}^2}}
- |\lambda| \bigg) \: d\tilde{\tau} \:.
\eeq

\begin{Lemma}\label{lemmaboundary_1}
For any~$\psi, \tilde{\psi} \in \H^\infty$ there exists a constant~$c>0$ such that for all~$m,m' \in I$,
\beq \label{decaytau}
\bigg| \int_{\R^3} \Sl \psi_m(\tau,x) | \tilde{\psi}_{m'}(\tau,x)\Sr\: R(\tau)^3\: d^3x \bigg| \leq \frac{c}{|\tau|}
\qquad \text{for all~$\tau < -1$}\:.
\eeq
\end{Lemma} 
\Proof According to \eqref{superpose}, \eqref{fundamental} and~\eqref{phikdef},
\beq \label{psiform}
\psi_m(\tau,x) = \int_{\R^3} \frac{d^3k}{(2 \pi)^3} \sum_{s=\pm} \sum_{a=1,2} \;
\hat{\psi}(m,k,s,a)\: R(\tau)^{-\frac{3}{2}}
\begin{pmatrix} u^{s,a}_1(m,k,\tau) \:\chi_{k, \pm} \\
u^{s,a}_2(m,k,\tau) \:\chi_{k, \pm} \end{pmatrix}\; e^{i k x} \:,
\eeq
and similarly for~$\tilde{\psi}_{m'}$.
Here the integrand is smooth and compactly supported. Therefore, the spatial integral in~\eqref{decaytau} can be computed with the help of Plancherel's theorem. Moreover, since for fixed~$k$ the spinors~$\chi_{k,s}$
are orthonormal, {\it cf.} \eqref{spinON}, we obtain
\begin{align}
&\int_{\R^3} \Sl \psi_m(t,x) | \tilde{\psi}_{m'}(t,x)\Sr\: d^3x
= \frac{1}{R^3} \int_{\R^3} \frac{d^3k}{(2 \pi)^3} \sum_{s=\pm} \sum_{a,a'=1,2} \notag \\
&\qquad \times \overline{\hat{\psi}(m,k,s,a)} \:\tilde{\psi}(m',k,s,a') \; \Big\la u^{s,a}(m,k,t), \begin{pmatrix} 1 & 0 \\ 0 & -1 \end{pmatrix} u^{s,a'}(m',k,t) \Big\ra_{\C^2} \:. \label{mixcombi}
\end{align}

Using the results of Lemma~\ref{lemmaboundary}, the contribution by $E(\tau)$ has the desired $1/\tau$-decay, see \eqref{Etau}, and it is smooth and compactly supported in~$k$. Therefore, using Plancherel's theorem, it satisfies the inequality~\eqref{decaytau}.

It remains to consider the first summand in~\eqref{uasybd}.
Since the operator~$U(\tau)$ in~\eqref{Udef} has the asymptotics
\[ U(\tau) = \frac{1}{\sqrt{2}}\;
\begin{pmatrix} 1 & s \\ -s & 1 \end{pmatrix} + \O\big(\tau^{-1} \big) \:, \]
it suffices to consider the constant matrix. Moreover, using that
\[ \begin{pmatrix} 1 & s \\ -s & 1 \end{pmatrix}^* 
\begin{pmatrix} 1 & 0 \\ 0 & -1 \end{pmatrix}
\begin{pmatrix} 1 & s \\ -s & 1 \end{pmatrix} =  \begin{pmatrix} 0 & 2s \\ 2s & 0 \end{pmatrix} \]
(where the star denotes transposition and complex conjugation),
in~\eqref{mixcombi} one only gets mixed contributions between the first component of~$g$
and the second one of~$\tilde{g}$ or vice versa.
Since all these terms can be treated in the same way, it suffices to consider one of them.
Therefore, the remaining task consists in showing that the integral
\begin{align*}
&\int_{\R^3} \frac{d^3k}{(2 \pi)^3} \;
\overline{\hat{\psi}(m,k,s,a)} \:\tilde{\psi}(m',k,s,a') \; \overline{g_1}\: \tilde{g}_2
\: e^{-i \varphi(\tau)- i \tilde{\varphi}(\tau)}
\end{align*}
decays like~$1/|\tau|$ for all~$s \in \{\pm\}$ and~$a, a' \in \{1,2\}$.
We point out that the phases~$\varphi(\tau)$ and~$\tilde{\varphi}(\tau)$ come with the same sign. This will play a crucial role below.

According to~\eqref{phichoice}, the function~$\varphi$ depends on~$\lambda = \pm |k|$.
We now integrate by parts as follows. First observe that
\begin{gather}
k^j \frac{\partial}{\partial k^j} \:\lambda = \pm k^j \:\frac{\partial}{\partial k^j} \:|k| = \pm |k| = \lambda \\
k^j \frac{\partial}{\partial k^j} \:\sqrt{\lambda^2 + m^2/\tau^2}
= \lambda\: \frac{\partial}{\partial \lambda} \:\sqrt{\lambda^2 + m^2/\tau^2}
= \frac{\lambda^2}{\sqrt{\lambda^2 + m^2/\tau^2}} \\
k^j \frac{\partial}{\partial k^j} \varphi(\tau) = \int_{-1}^\tau \frac{\lambda^2}{\sqrt{\lambda^2 + m^2/\eta^2}}\:
d\eta =: \rho(k,\tau) \:. \label{rhodef}
\end{gather}
Hence
\begin{align*}
&\int_{\R^3} \frac{d^3k}{(2 \pi)^3} \;
\overline{\hat{\psi}(m,k,s,a)} \:\tilde{\psi}(m',k,s,a')
\: e^{-i \varphi(\tau)- i \tilde{\varphi}(\tau)} \\
&= \int_{\R^3} \frac{d^3k}{(2 \pi)^3} \;
\overline{\hat{\psi}(m,k,s,a)} \:\tilde{\psi}(m',k,s,a')
\;\frac{i}{\rho(k,\tau) + \tilde{\rho}(k,\tau)}\;k^j \frac{\partial}{\partial k^j} 
\: e^{-i \varphi(\tau)- i \tilde{\varphi}(\tau)}.
\end{align*}
Now we integrate the $k$-derivatives by parts.

The resulting contributions are estimated as follows: Expanding the integral in~\eqref{rhodef} for large~$|\tau|$, we obtain
\beq \label{rhoes}
\rho(\tau,k) = |k|\: |\tau| + \O\big(\tau^{-1} \big) \:.
\eeq
Therefore, when the $k$-derivative acts on the factors~$k^j$, $\hat{\psi}(m,k,s,a)$
or~$\chi_{k,s}(x)$, the resulting contributions decay like~$1/\tau$, while staying
smooth and compactly supported in~$k$. Therefore we establish the
desired bound. If the $k$-derivative acts on the factor~$\rho$ (and similarly on~$\tilde{\rho}$), we obtain
\[ k^j \,\frac{\partial}{\partial k^j} \rho(\tau,k) = \lambda\: \frac{\partial}{\partial \lambda} \rho(\tau,k)
= |\tau| + \O\big(\tau^{-1} \big) \:, \]
which together with~\eqref{rhoes} again gives the desired~$1/|\tau|$-decay.

It remains to be shown that the decay~$\sim 1/|\tau|$ is uniform in~$k$.
This follows from the fact that, according to~\eqref{Hinf}, the elements in~$\H^\infty$
have compact support in~$k$ away from ~$k=0$. This concludes the proof.
\QED

Combining Lemma~\ref{lemmainfty} and Lemma~\ref{lemmaboundary},
we immediately obtain the following result:
\begin{Prp} \label{prpdouble}
For any~$\psi, \phi \in \H^\infty$,
\beq \label{double}
\int_{-\infty}^\infty dt \;\bigg| \int_{\R^3} \Sl \p \psi | \p \phi \Sr_{(t,x)}\: R(t)^3\: d^3x \bigg| < \infty \:.
\eeq
\end{Prp}
\Proof Asymptotically as~$t \rightarrow \infty$, it was shown in Lemma~\ref{lemmainfty} 
that $\p \psi$ decays at least like~$1/\tau$.
Therefore, it is square integrable over~$t \in [1,\infty)$.
Near the boundary, on the other hand, it was shown in Lemma~\ref{lemmaboundary}
that the spatial integral of the inner product~$\Sl \psi_m | \psi_{m'} \Sr$
decays like~$1/|\tau|$. In view of the fact that the integration measure~$dt$
transforms to~$d\tau/|\tau|$, this implies that the time integral also
exists in the $L^1$-sense near the boundary.
Keeping in mind that the estimates are locally uniform in the mass parameters
and that~$m,m' \in I$ with~$I$ a compact interval,
we obtain the sought result.
\QED

\begin{Remark} {\em{ We point out that,
in contrast to the assumptions in the strong mass oscillation property,
here the function~$\Sl \p \psi | \p \phi \Sr$ will in general
not be integrable on~$\scrM$. We now explain how this comes about.
Using the asymptotics in~\eqref{uasybd} in~\eqref{psiform}, one finds that
the leading contribution to~$\psi$ near the boundary is of the form
\beq \label{leading}
\phi(t,x) := R(t)^{\frac{3}{2}}\: \psi(t,x) \sim \int_{\R^3} \frac{d^3k}{(2 \pi)^3}\: 
h(k)\: e^{i k x}\:e^{\pm i \,|k|\, t} \:.
\eeq
By direct computation, one sees that~$\phi$ is a solution of the scalar wave equation
in Minkowski space. Therefore, it can
be represented by a Fourier integral of the form
\[ \phi(t,x) = \int_{\R^4} \frac{d^4p}{(2 \pi)^4}\: \hat{\phi}(p)\: \delta\big( p^2 \big)\: e^{-i (p^0 t - p^1 x)}\:. \]
Consequently, using Plancherel's theorem, its spatial $L^2$-norm
\[ \int_{\R^3} |\psi(t,x)|^2\: d^3x = \sum_\pm \int_{\R^3} \frac{d^3k}{(2 \pi)^3}\: \frac{1}{2\,|k|}\:
\big|\hat{\psi}\big( \pm |k|, k \big) \big|^2 \]
is time independent. Therefore, the integral
\[ \int_{-\infty}^0 \frac{d\tau}{|\tau|} \int_{\R^3}d^3x\: |\psi(t,x)|^2 \]
diverges.

One may wonder whether integrating over the mass parameter might improve the
decay properties due to mass oscillations. This is not the case, as the following consideration shows:
Differentiating~\eqref{phichoice} with respect to the mass, one finds that
\[ \int_{-\infty}^{-1} \Big| \frac{\partial \phi(\tau)}{\partial m} \Big| \:d\tau < \infty \:. \]
Therefore, varying the mass only gives rise to a finite phase shift in~\eqref{uasybd},
implying that mass oscillations do not improve the decay properties.

Finally, one may ask whether the fact that the combination~$|\Sl \psi_m | \phi_{m'} \Sr|$
might become integrable due to the fact that the inner product is indefinite.
This is also not the case, as the following argument shows: 
As it becomes apparent in~\eqref{stip3tau}, the leading contribution~\eqref{leading}
does enter the inner product. In~\eqref{stip3tau}, we made use of phase factors to
make sense of the time integral. However, if the absolute value of the inner
product~$\Sl \psi_m | \phi_{m'} \Sr$ is taken, the phase information gets lost and the
time integral necessarily diverges.

These considerations show that the spin scalar product~$\Sl \p \psi | \p \phi\Sr$
is in general {\em{not integrable}}. With the procedure~\eqref{double}
we can make sense of the spacetime integral, but only if we integrate
{\em{first over space}} and {\em{then over}} the {\em{time}} variable.

The critical reader may wonder how our findings fit together with the results in~\cite[Section~6]{infinite}.
Indeed, in~\cite{infinite} it was shown that the function~$\Sl \p \psi | \p \phi \Sr$ is integrable over the whole
de Sitter, which clearly implies that this function is also integrable on any of its open subsets, in particular the cosmological de Sitter spacetime. In order to understand why these results do not contradict each other,
one must keep in mind that the domains~$\H^\infty$ were chosen differently. In particular, in~\cite[Section~6]{infinite} the elements in $\H^\infty$
do not decay rapidly on the surfaces~$t=\text{const}$
in cosmological de Sitter. In the analysis in this paper, however, working with rapidly decaying solutions
(more specifically, with compactly supported solutions in momentum space) is easier
and more natural. The prize we pay is that the function~$\Sl \p \psi | \p \phi \Sr$ need not be
in~$L^1(\scrM)$.
}} \hspace*{1cm} \QEDrem
\end{Remark}

Next we establish a relation similar to~\eqref{mb2}.
Here there are two major differences: First, the spacetime integral must be
defined as in~\eqref{double} by first integrating over space and then over time.
Second, integrating the Dirac operator by parts gives boundary terms
denoted by~${\mathfrak{B}}$.
\begin{Prp} \label{prpT}
For any~$\psi, \tilde{\psi} \in \H^\infty$,
\begin{align*}
&\int_{-\infty}^\infty dt \;\bigg( \int_{\R^3} \big( 
\Sl \p T \psi \,|\, \p \tilde{\psi} \Sr_{(t,x)} - \Sl \p \psi \,|\, \p T \tilde{\psi} \Sr_{(t,x)} \big) \: R(t)^3\: d^3x \bigg) \\
&\;= i \int_I dm \int_I dm'\; {\mathfrak{B}} \big( \psi_m, \tilde{\psi}_{m'} \big) \:,
\end{align*}
where
\[ {\mathfrak{B}}(\psi_m, \tilde{\psi}_{m'}) := \overline{g_1} \tilde{g_1} + \overline{g_2} \tilde{g_2} \,, \]
and $g_{1\!/\!2}$ and~$\tilde{g}_{1\!/\!2}$ are the coefficients in the asymptotic expansion of Lemma~\ref{lemmaboundary} with~$\varphi(\tau)$ according to~\eqref{phichoice}.
\end{Prp}
\Proof Using that~$\psi_m$ and~$\tilde{\psi}_m$ satisfy the Dirac equation, we may replace
the factors~$T$ by the Dirac operator~$\Dir$. As the latter is symmetric with
respect to the inner product~$\bra .|. \ket$, we only need to compute the boundary terms.
Since for every fixed time, the wave functions decay rapidly at spatial infinity, the spatial part
of the Dirac operator does not give rise to boundary terms.
Moreover, the decay properties as~$t \rightarrow +\infty$ as worked out in Lemma~\ref{lemmainfty}
imply that we do not get boundary terms at infinity. Therefore, it remains to compute the boundary terms as~$t \rightarrow -\infty$. On account of the time component of the metric~\eqref{flat_dS_metric}, this is 
\begin{align*}
A \,\,&\!\!:= \int_{-\infty}^\infty dt \:\bigg( \int_{\R^3} \big( 
\Sl \p T \psi \,|\, \p \tilde{\psi} \Sr_{(t,x)} - \Sl \p \psi \,|\, \p T \tilde{\psi} \Sr_{(t,x)} \big) \: R(t)^3\: d^3x \bigg) \\
&= \lim_{t \rightarrow -\infty}
\int_{\R^3} \Sl \p \psi \,|\, \p \,\big(i \gamma^0)\, \tilde{\psi} \Sr_{(t,x)} \: R(t)^3\: d^3x \:.
\end{align*}
As in~\eqref{mixcombi}, the spatial integral can be carried out using Plancherel's theorem. Therefore we obtain
\begin{align}
A &= \lim_{t \rightarrow -\infty} \int_{\R^3} \frac{d^3k}{(2 \pi)^3} \sum_{s=\pm} \sum_{a,a'=1,2} \notag \\
&\qquad\qquad \times \overline{\hat{\psi}(m,k,s,a)} \:\tilde{\psi}(m',k,s,a') \;
i \: \big\la u^{s,a}(m,k,t), u^{s,a'}(m',k,t) \big\ra_{\C^2} \:. \label{Aform}
\end{align}
Using the asymptotics of Lemma~\ref{lemmaboundary}, we only need to take into account
the plane wave asymptotics in~\eqref{uasybd}. Since the operator~$U(\tau)$ is unitary,
it drops out from the scalar product in~\eqref{Aform}.
Moreover, in the limit~$\tau \rightarrow -\infty$, the integral in~\eqref{phichoice} vanishes.
Therefore, we can work with the simple formula~$\phi(\tau) = |\lambda|\, \tau$, implying that
the phases drop out of the scalar product in~\eqref{Aform}. This concludes the proof.
\QED

It is worth noting that the boundary terms are positive in the sense that
\[ 
{\mathfrak{B}}(\psi_m, \psi_m) \geq 0 \:. \]
This can be understood from the fact that the boundary terms tell about the
flux of the electromagnetic current through the null surface~$t=-\infty$.
The Dirac current flux through a null surface always has a definite sign.

\begin{Thm} {\bf{(mass decomposition)}} \label{thmdecomp}
For all~$\psi, \phi \in \H^\infty$,
\begin{align}
\int_{-\infty}^\infty &dt \;\bigg( \int_{\R^3} \Sl \p \psi | \p \phi \Sr_{(t,x)} \: R(t)^3\: d^3x \bigg) \notag \\
&= \int_I (\psi_m \,|\, \Sig_m \,\phi_m)_m\: dm \label{Sigdef} \\
&\quad\; + i \,\lim_{\varepsilon \searrow 0}
\int_I dm \int_I dm'\; \frac{m-m'}{(m-m')^2+\varepsilon^2} \: {\mathfrak{B}}\big(\psi_m, \phi_{m'} \big) \:, \label{Bdef}
\end{align}
where~$\Sig_m$ is the operator which in the separation ansatz~\eqref{ansatzhalf}
acts on the fundamental solutions~$u^{s,a}$ by
\beq \label{Smhalf}
\Sig_m u^{s,1} = \frac{1}{2}\: u^{s,1} \:, \qquad \Sig_m u^{s,2} = -\frac{1}{2}\: u^{s,2} \:.
\eeq
\end{Thm}
\Proof As worked out in Lemma~\ref{lemmainfty}, the $t$-integral for large times
exists only as a consequence of the mass oscillations. In order to compute the integrals,
it is convenient to insert a convergence-generating function~$\eta_\varepsilon(t)$ defined
for given~$L \in \R$ by
\[ \eta_\varepsilon(t) = \chi_{(-\infty, L]}(t) + e^{-\varepsilon t}\: \chi_{(L, \infty)}(t) \:. \]
Then, using Lebesgue dominated convergence theorem,
\begin{align}
A(\psi,\phi) \,\,&\!\!:= \int_{-\infty}^\infty dt \;\bigg( \int_{\R^3} \Sl \p \psi | \p \phi \Sr_{(t,x)} \: R(t)^3\: d^3x \bigg) \notag \\
&= \lim_{\varepsilon \searrow 0} \int_{-\infty}^\infty \eta_\varepsilon(t) \:
dt \;\bigg( \int_{\R^3} \Sl \p \psi | \p \phi \Sr_{(t,x)} \: R(t)^3\: d^3x \bigg) \notag \\
&= \lim_{\varepsilon \searrow 0} \int_I dm \int_I dm' \int_{-\infty}^\infty \eta_\varepsilon(t) \:
dt \;\bigg( \int_{\R^3} \Sl \psi_m \,|\, \phi_{m'} \Sr_{(t,x)} \: R(t)^3\: d^3x \bigg) \:. \label{totint}
\end{align}

First we compute the contributions for~$m \neq m'$. To this end, we assume that~$\psi$
and~$\phi$ are supported in disjoint subsets of~$I$. Then Proposition~\ref{prpT} implies that
\begin{align*}
&\lim_{\varepsilon \searrow 0} \int_I dm \int_I dm' \int_{-\infty}^\infty \eta_\varepsilon(t) \:
dt \;\bigg( \int_{\R^3} \:(m-m')\: \Sl \psi_m \,|\, \phi_{m'} \Sr_{(t,x)} \: R(t)^3\: d^3x \bigg) \\
&= i \int_I dm \int_I dm'\; {\mathfrak{B}} \big( \psi_m, \phi_{m'} \big).
\end{align*}
Since~$\psi$ and~$\phi$ may be multiplied by arbitrary test functions in~$m$
(see Definition~\ref{defHinf}~(i)), it follows that
\[ A(\psi,\phi) = i \int_I dm \int_I dm'\; \frac{{\mathfrak{B}} \big( \psi_m, \phi_{m'} \big)}{m-m'} \:. \]
We have thus derived the contribution for~$m \neq m'$ in~\eqref{Bdef}.

It remains to compute the singular contribution at~$m=m'$.
According to Lemma~\ref{lemmaboundary}, the spacetime integral in~\eqref{totint}
exists for~$t<L$ pointwise in~$m$ and~$m'$ and is uniformly bounded in~$m$ and~$m'$.
Therefore, it suffices to consider the integrals
\begin{align*}
\int_I dm \int_I dm' \int_L^\infty e^{-\varepsilon t} \:
dt \;\bigg( \int_{\R^3} \Sl \psi_m \,|\, \phi_{m'} \Sr_{(t,x)} \: R(t)^3\: d^3x \bigg)
\end{align*}
in the limit~$\varepsilon \searrow 0$. Using the asymptotics of Lemma~\ref{lemmagronwall},
the contribution by the error terms in~\eqref{eq:3z} can be made arbitrarily small by increasing~$L$.
Therefore, it suffices to consider the plane waves in~\eqref{eq:3z}.
Employing the ansatz~\eqref{superpose} and~\eqref{fundamental},
we can again carry out the spatial integral with Plancherel's theorem. We obtain the integral
\begin{align*}
&\int_I dm \int_I dm' \int_L^\infty e^{-\varepsilon t} \:dt
\int_{\R^3} \frac{d^3k}{(2 \pi)^3} \sum_{s=\pm} \sum_{a,a'=1,2} \;
\overline{\hat{\psi}(m,k,s,a)}\: \hat{\phi}(m',k,s,a') \\
&\qquad \times \Sl \begin{pmatrix} e^{-imt} \:f^\infty_1(m,k,s,a) \\[0.3em]
e^{imt} \:f^\infty_2(m,k,s,a) \end{pmatrix} \:|\: \begin{pmatrix} e^{-im't} \:\tilde{f}^\infty_1(m',k,s,a') \\[0.3em]
e^{im't} \:f^\infty_2(m',k,s,a') \end{pmatrix} \Sr \\
&=\int_I dm \int_I dm'
\int_{\R^3} \frac{d^3k}{(2 \pi)^3} \sum_{s=\pm} \sum_{a,a'=1,2} \;
\overline{\hat{\psi}(m,k,s,a)}\: \hat{\phi}(m',k,s,a') \\
&\qquad \times \Big( \frac{\delta^a_1\: \delta^{a'}_1}{-i (m-m') + \varepsilon}\:
e^{i(m-m')L - \varepsilon L}
- \frac{\delta^a_2\: \delta^{a'}_2}{i (m-m') + \varepsilon} \Big) \: e^{-i(m-m')L - \varepsilon L}\:,
\end{align*}
where in the last step we carried out the $t$-integration and
used the asymptotics of the fundamental solutions~\eqref{uasy}.
Taking the limit~$\varepsilon \searrow 0$ with the help of the distributional relation
\[ \lim_{\varepsilon \searrow 0} \frac{1}{x \pm i \varepsilon} =
\mp i \pi \: \delta(x) + \frac{\text{PP}}{x} \]
(where~$\text{PP}$ denotes the principal part),
we find that the singular contribution at~$m=m'$ consists of a $\delta$-distribution
and a principal part. The contribution by the former to~$A(\psi,\phi)$ is
\begin{align*}
& \int_I dm \int_I dm'
\int_{\R^3} \frac{d^3k}{(2 \pi)^3} \sum_{s=\pm} \sum_{a,a'=1,2} \;
\overline{\hat{\psi}(m,k,s,a)}\: \hat{\phi}(m',k,s,a') \\
&\qquad \times \pi \: \delta(m-m')\: \big( \delta^a_1\: \delta^{a'}_1 - 
\delta^a_2\: \delta^{a'}_2 \big) \\ 
&= \pi \int_I dm
\int_{\R^3} \frac{d^3k}{(2 \pi)^3} \sum_{s=\pm} \sum_{a=1,2} \;
\overline{\hat{\psi}(m,k,s,a)}\: \hat{\phi}(m,k,s,a) \; \epsilon(a) \:.
\end{align*}
Comparing this formula with the expression for the scalar product~\eqref{print3h}
evaluated asymptotically as~$t \rightarrow \infty$, one obtains~\eqref{Sigdef} with~$\Sig_m$ given by~\eqref{Smhalf}. This concludes the proof.
\QED

We remark that a similar connection between boundary terms and double
mass integrals involving a principal value has already been discovered in the analysis of the fermionic signature operator in the exterior region of Schwarzschild spacetime~\cite{FR}.

\subsection{The Fermionic Projector State}

Starting from Theorem \ref{thmdecomp} and from the operator $\Sig_m$ identified in \eqref{Sigdef}, we can construct first of all the fermionic projector $P$ by \eqref{Pdef}, that is $P=-\chi_{(-\infty, 0)}(\Sig_m)k_m$. Hence following Theorem \ref{Thm:2-pt} and Theorem \ref{thmstate}, we have identified a quasi-free state for the algebra of smeared fields whose two-point function reads
\begin{equation}\label{Eq:2-pt-half}
\omega(\Psi(h)\Psi^*(f))=-\mathcal{P}(h\otimes f)=-\langle h|Pf\rangle \:.
\end{equation}

We observe two important features of the state associated to $\omega$, which will help unveiling its physical significance:
\begin{enumerate}[leftmargin=2.5em]
\item Since $\omega$ is constructed out of the causal fundamental solution $k_m$ and out of the fermionic signature operator $\Sig_m$ defined in \eqref{Smdef}, the two-point correlation function is invariant under the action of all background isometries. For more details refer to \cite{Finster:2017wco} and to Lemma \ref{Lem:isometry_invariance}. In the case of \eqref{flat_dS_metric}, it corresponds to the three-dimensional Euclidean group $E(3)$, which encodes the information that each Cauchy surface at constant $t$ is isometric to the three dimensional Euclidean space. In other words, following Lemma \ref{Lem:isometry_invariance}, for every $\alpha:\scrM\to\scrM$ identifying an element of the isometry group Iso$(\scrM,g)$ of $(\scrM,g)$, it turns out that
	$$\alpha^*(\mathcal{P})(h\otimes f)=\mathcal{P}(h\otimes f).$$
\item From \eqref{Smhalf} and \eqref{uasy} we can infer that, in the limit $t\to\infty$, the projection on the negative spectral subspace in the definition of $\mathcal{P}$ entails that only the component proportional to $u_1^{s,a}$ plays a role in the two-point correlation function.
\end{enumerate}

One might wonder whether the fermionic projector state coincides with the restriction to the cosmological de Sitter spacetime of the spinorial Bunch-Davies state in the sense of Section \ref{Sec:Hadamard_states}. Since we are considering only an open subset of full de Sitter spacetime, we cannot apply directly the results of Allen \& L\"utken \cite{Allen} and we need to rely on a different procedure aimed at constructing maximally symmetric states in a cosmological spacetime with flat spatial sections, \cite{Dappiaggi:2010gt}. Without entering into the details of this construction, which would bring us far from the scopes of this paper, we remark that 
the procedure of \cite{Dappiaggi:2010gt} is based on the observation that the two-point function of the spinorial Bunch-Davies state selects only negative frequencies on the conformal boundary $\mathscr{I}^-$ corresponding to $\tau\to -\infty$. Observe that, on $\mathscr{I}^-\simeq\mathbb{R}\times\mathbb{S}^2$, one can define coherently a notion of frequency with respect to the rigid translations along the $\mathbb{R}$-direction. The associated generator is nothing but the push-forward to $\mathscr{I}^-$ of $\partial_\tau$, $\tau$ corresponding to the conformal time. 

As a consequence of this observation, it follows that the state built via the two-point function \eqref{Eq:2-pt-half} constructed out of the fermionic signature operator \eqref{Sigdef} cannot coincide with that of the spinorial Bunch-Davies state. More precisely, combining \eqref{Smhalf} with the fact that \eqref{Eq:2-pt-half} is a bi-solution of the Dirac equation, it turns out that each negative frequency mode at $\tau\to-\infty$ must evolve via \eqref{Diractau}. Such equation entails that, at $\tau\to 0$, also positive frequencies appear in the mode decomposition, in contradiction to the analysis of \cite{Dappiaggi:2010gt}.

\section{De Sitter in Closed Slicing}\label{Sec:closed_dS}
\subsection{Preliminaries to De Sitter}
In this section we consider de Sitter spacetime in the so-called closed slicing, {\it cf.} \cite{Moschella:2006pkh}. We also recall a few results from~\cite[Section~6]{infinite} using the notation of this paper. The underlying background is $\scrM=\R \times \mathbb{S}^3$ with  line element
\[ 
ds^2 = dT^2 - R(T)^2\: ds^2_{ \mathbb{S}^3} \qquad \text{and} \qquad R(T) = \cosh T \:, \]
where~$ds^2_{ \mathbb{S}^3}$ is the line element of the three-dimensional unit sphere.
This is a special case of a Friedmann-Robertson-Walker metric with closed spatial sections. The Dirac operator was computed in~\cite{moritz} to be
\[ 
\Dir = i \gamma^0 \left( \partial_T + \frac {3 \dot{R}(T)}{2R(T)} \right) + \frac{1}{R(T)}
\begin{pmatrix} 0 & \Dir_{ \mathbb{S}^3} \\ -\Dir_{ \mathbb{S}^3} & 0 \\ \end{pmatrix} , 
\]
where~$\Dir_{ \mathbb{S}^3}$ is the Dirac operator on~$ \mathbb{S}^3$. The space time inner product \eqref{stip}
and the scalar product \eqref{print} take the form
\begin{align}
\bra \psi | \phi \ket &= \int_{-\infty}^\infty dT \int_{ \mathbb{S}^3} \Sl \psi | \phi \Sr_{(T,x)} \:R(T)^3\, d\mu_{ \mathbb{S}^3}(x) \label{stip2} \\
( \psi_m | \phi_m )_m &= 2 \pi \int_{ \mathbb{S}^3} \Sl \psi | \gamma^0 \phi \Sr_{(T,x)} \:R(T)^3\,  d\mu_{ \mathbb{S}^3}(x) \:, \label{print2}
\end{align}
where~$\Sl \psi | \phi \Sr = \psi^\dagger \gamma^0 \phi$ and~$\gamma^0=\diag(1,1,-1,-1)$, while $d\mu_{ \mathbb{S}^3}$ is the normalized volume measure on~$ \mathbb{S}^3$.

The Dirac equation can be separated with the ansatz
\[ 
\psi_m = R(T)^{-\frac{3}{2}}
\begin{pmatrix} u_1(m,T) \:\phi^{(\lambda)}(x) \\
u_2(m,T) \:\phi^{(\lambda)}(x) \end{pmatrix} \:, \]
where~$\phi^{(\lambda)}(x)$ is a normalized eigenspinor~$\phi^{(\lambda)}$ of~$\Dir_{ \mathbb{S}^3}$ corresponding
to the eigenvalue~$\lambda \in \{ \pm \frac{3}{2}, \,\pm \frac{5}{2}, \,\pm \frac{7}{2}, \ldots \}$.
The resulting ODE in time takes the form
\beq \label{DiracODE}
i\, \frac{d}{dT} \begin{pmatrix} u_1 \\ u_2 \end{pmatrix}
=  \begin{pmatrix} m & -\lambda/R \\ -\lambda/R & - m \end{pmatrix}
\begin{pmatrix} u_1 \\ u_2 \end{pmatrix}
\eeq
for the complex-valued functions~$u_1$ and~$u_2$.
Using asymptotic estimates for the solution of this ODE, in~\cite[Section~6]{infinite} it is shown
that the Dirac operator has the strong mass oscillation property, as we now recall.
We decompose the solution space into spatial modes,
\[ \H_m = \bigoplus_{\lambda \in \sigma(\Dir_{ \mathbb{S}^3})} \H_m^{(\lambda)}\:,\qquad
\H = \bigoplus_{\lambda \in \sigma(\Dir_{ \mathbb{S}^3})} \H^{(\lambda)} \:, \]
and we introduce the domain~$\H^\infty$ as the solutions composed only of a finite
number of modes,
\beq \label{Hinfchoice2}
\H^\infty = \Big\{ \psi \in \Cisco(\scrM \times S, S\scrM) \cap \H \;\Big|\;
\psi \in \bigoplus\nolimits_{|\lambda| \leq \Lambda} \H^{(\lambda)} \text{ with } \Lambda \in \R \Big\} \:.
\eeq
\begin{Thm} \label{thmdeSitter} On any interval~$I=(m_L, m_R)$ with~$m_L, m_R>0$,
the Dirac operator in de Sitter spacetime enjoys the strong mass oscillation property with domain~\eqref{Hinfchoice2}.
\end{Thm}

For a single spatial mode, the inner products~\eqref{stip2} and~\eqref{print2} become
\begin{align}
\bra \psi | \tilde{\psi} \ket &= \int_{-\infty}^\infty \left(\overline{u_1} \tilde{u}_1 - \overline{u_2} \tilde{u}_2 \right) dT  \label{stip3} \\
( \psi_m \,|\, \tilde{\psi}_m )_m &= 2 \pi \left( \overline{u_1} \tilde{u}_1 + \overline{u_2} \tilde{u}_2 \right) 
= 2 \pi \la u, \tilde{u} \ra_{\C^2} \:. \label{print3}
\end{align}

\subsection{The Fermionic Signature Operator}
In~\cite[Section~6]{infinite} the fermionic signature operator was
derived and computed. We now recall a few results of this analysis which will be of relevance here.
First, in~\cite[Lemma~6.2]{infinite} the asymptotics of the solutions of the ODE~\eqref{DiracODE}
is determined. The result is very similar to the asymptotics in Lemma~\ref{lemmagronwall},
but now the plane wave asymptotics is obtained both in the future and in the past:
\begin{Lemma} \label{lemma31}
Asymptotically as $T \rightarrow \pm \infty$, every solution of~\eqref{DiracODE} is of the form
\beq u(T) = \begin{pmatrix} e^{-imT} \:f^\pm_1 \\[0.3em]
e^{imT} \:f^\pm_2 \end{pmatrix} + E^\pm(T) \label{eq:3z}
\eeq
with the error term bounded by
\[ 
\|E^\pm(T)\| \leq \|f^\pm\| \, \exp \big( 2 \,|\lambda| \,e^{\mp T} \big) \:, 
\]
hence entailing exponential decay of $E^\pm(T)$ as~$T \rightarrow \pm \infty$).
\end{Lemma}

\noindent In addition \cite[Lemma~6.6]{infinite} entails the following:
\begin{Lemma} \label{lemmamode}
For any single modes~$\psi, \tilde{\psi} \in \H^{(\lambda)}$ with the same spatial dependence,
\beq \label{modeform}
\bra \p \psi | \p \tilde{\psi} \ket = \pi \sum_{s = \pm} \int_I \left( \overline{f^s_1(m)} \tilde{f}^s_1(m)
- \overline{f^s_2(m)} \tilde{f}^s_2(m) \right) dm \:.
\eeq
\end{Lemma}

Comparing~\eqref{modeform} with the representation of the scalar product~\eqref{print3}
(which can be evaluated asymptotically as~$T \rightarrow \pm \infty$), one can immediately read off
the fermionic signature operator:
\begin{Corollary}\label{fermionic_signature}
\[ \bra \p \psi | \p \tilde{\psi} \ket = \int_I (\psi_m \,|\, \Sig_m \,\phi_m)_m\: dm \:, \]
where the operator~$\Sig_m$ is of the form
\[ \Sig_m = \frac{1}{2} \:\big( \Sig_m^+ + \Sig_m^- \big) \]
with operators~$\Sig_m^\pm$ which modify the asymptotics of the solutions
in~\eqref{eq:3z} to
\begin{equation}\label{eq:error_term}
\big(\Sig_m^\pm u\big) (T) = \begin{pmatrix} e^{-imT} \:f^\pm_1 \\[0.3em]
-e^{imT} \:f^\pm_2 \end{pmatrix} + \tilde{E}^\pm(T) \qquad \text{as~$T \rightarrow \pm \infty$}\:,
\end{equation}  
where the error term again decays exponentially~\eqref{eq:3c}.
\end{Corollary}

In the remainder of this paper, we shall compute the fermionic signature operator
and the resulting fermionic projector state in more detail.
To this end, we shall make use of the fact that the ODE~\eqref{print3}
has explicit solutions in terms of hypergeometric functions.
We compute these explicit solutions by first considering the corresponding
second order scalar ODEs.

\subsection{Derivation of Second Order Scalar Equations}
It is convenient to deduce from the Dirac equation~\eqref{DiracODE} two second order
scalar equations for~$u_1$ and~$u_2$. To this end, one multiplies
the first equation in~\eqref{DiracODE} by~$e^{-imT} \,R(T)$ and the second one
by~$e^{imT} \,R(T)$. Subsequently one differentiates both equations with respect to~$T$, obtaining
\begin{align*}
\ddot{u}_1 &= -m^2\, u_1 - \frac{i m \dot{R}}{R}\: u_1 - \frac{\dot{R}}{R}\: \dot{u}_1 
+ \frac{\lambda}{R}\: \big( i \dot{u}_2 + m u_2 \big)\:, \\
\ddot{u}_2 &= -m^2\, u_2  + \frac{i m \dot{R}}{R}\: u_2 - \frac{\dot{R}}{R}\: \dot{u}_2
+ \frac{\lambda}{R}\: \big( i \dot{u}_1 - m u_1 \big) \:,
\end{align*}
which can be decoupled by inserting the second and first equation in~\eqref{DiracODE},
respectively. We obtain
\begin{align}
\ddot{u}_1 &= -m^2\, u_1 - \frac{\lambda^2}{R^2}\: u_1 
- \frac{i m \dot{R}}{R}\: u_1 - \frac{\dot{R}}{R}\: \dot{u}_1 \label{u1eq} \\
\ddot{u}_2 &= -m^2\, u_2  - \frac{\lambda^2}{R^2}\: u_2
+ \frac{i m \dot{R}}{R}\: u_2 - \frac{\dot{R}}{R}\: \dot{u}_2 \:. \label{u2eq}
\end{align}

\subsection{Explicit Solution of the Dirac Equation}
Employing the ansatz
\[ u_1(T) = e^{-i m T} \: v_1 (z) \qquad \text{with} \qquad z(T) = \frac{e^{-T}}{e^{T} + e^{-T}} \:, \]
\eqref{u1eq} transforms to
\[ z(1-z)\: v_1''(z) + \Big(\frac{1}{2} + i m -z \Big)\, v_1'(z) + \lambda^2\: v_1(z) = 0 \:. \]
This is the hypergeometric differential equation, see e.g.~\cite[eqs~15.10.1]{dlmf},
having as special solution the hypergeometric series~${}_{2}F_{1}(-\lambda, \lambda; \frac{1}{2}+im; z)$.
Therefore we obtain a particular solution of~\eqref{u1eq}
\beq \label{up1}
u^+_1(T) = e^{-i m T} \: {}_{2}F_{1}\Big(-\lambda, \lambda; \frac{1}{2}+im; \frac{e^{-T}}{e^{T} + e^{-T}} \Big)  \:.
\eeq
Substituting it into the first Dirac equation in~\eqref{DiracODE}, one can solve for~$u_2$,
\begin{align}
u^+_2(T) &= \frac{\cosh T}{\lambda} \,\Big( - i \dot{u}^+_1(T) + m \,u^+_1(T) \Big) \notag \\
&= -\frac{2}{2m-i}\: \frac{e^{-i m T}}{e^{T} + e^{-T}} \; {}_{2}F_{1}\Big(1-\lambda, 1+\lambda; \frac{3}{2}
+ im; \frac{e^{-T}}{e^{T} + e^{-T}} \Big) \:, \label{up2}
\end{align}
where in the last step we used the formula for the derivatives of hypergeometric
functions~\cite[eq.~15.5.1]{dlmf}
\beq \label{hyperdiff}
\frac{d}{dz}\, {}_{2}F_{1} \big(a,b;c;z \big) =\frac{ab}{c}\;
{}_{2}F_{1} \big(a+1,b+1;c+1;z \big)\:.
\eeq
Similarly, for~\eqref{u2eq} we obtain the fundamental solution
\beq \label{um2}
u^-_2 = e^{i m T} \: {}_{2}F_{1}\Big(-\lambda, \lambda; \frac{1}{2}-im; \frac{e^{-T}}{e^{T} + e^{-T}} \Big) \:.
\eeq
Substituting it into the second Dirac equation in~\eqref{DiracODE} and solving for~$u_1$,
we obtain
\begin{align}
u^-_1(T) &= \frac{\cosh T}{\lambda} \,\Big( - i \dot{u}^-_2(T) - m \,u^-_2(T) \Big) \notag \\
&= -\frac{2}{2m+i}\: \frac{e^{i m T}}{e^{T} + e^{-T}} \; {}_{2}F_{1}\Big(1-\lambda, 1+\lambda; \frac{3}{2}
- im; \frac{e^{-T}}{e^{T} + e^{-T}} \Big) \:, \label{um1}
\end{align}
where in the last step we applied again \eqref{hyperdiff}.

With~\eqref{up1}, \eqref{up2} and~\eqref{um2}, \eqref{um1} we have constructed
two solutions of the Dirac equation~\eqref{DiracODE}. These form a fundamental system. This can be seen most easily by noting that they have a different asymptotics as~$T \rightarrow \infty$, as it will be worked out in detail in the next section.

\subsection{Asymptotics of the Solutions}\label{Sec:asymp_sol}
We evaluate the asymptotics of the Dirac solutions as~$T \rightarrow \pm \infty$.
We begin with that as~$T \rightarrow +\infty$.
In this limiting case, the last argument of the hypergeometric functions
in~\eqref{up1}, \eqref{up2}, \eqref{um2}, \eqref{um1} tends to zero, making it
possible to use the power expansion (see~\cite[eqs~15.2.1]{dlmf})
\beq \label{zzero}
{}_{2}F_{1}\big(a,b;c;z \big) = 1 + \O(z) \:.
\eeq
We thus obtain the simple asymptotics
\[ 
\begin{split} 
\begin{pmatrix} u_1^+ \\[0.3em] u_2^+ \end{pmatrix}
&= \begin{pmatrix} e^{-i mT} \\ 0 \end{pmatrix} + \O \big( e^{-T} \big) \\
\begin{pmatrix} u_1^- \\[0.3em] u_2^- \end{pmatrix}
&= \begin{pmatrix} 0 \\ e^{i mT} \end{pmatrix} + \O \big( e^{-T} \big) \:, 
\end{split} 
\]
which also shows that our two solutions of the Dirac equation are linearly independent and thus
form a fundamental system of~\eqref{DiracODE}.

The asymptotics as~$T \rightarrow -\infty$ is a bit more difficult because, in this
limit, the last argument of the hypergeometric functions
in~\eqref{up1}, \eqref{up2}, \eqref{um2}, \eqref{um1} tends to one.
Even if ${}_{2}F_1$ is the standard notation for hypergeometric functions, it is more convenient to work with the function
$$\mathbf{F}\left({a,b\atop c};z\right)  := \frac{{{}_{2}F_{1}}\left(a,b;c;z\right)}{\Gamma\left(c\right)}\,,$$
where $\Gamma$ is Euler's Gamma function.
Applying the relation between hypergeometric functions~\cite[eq~ 15.8.4]{dlmf}
\begin{align*}
\frac{\sin\left(\pi(c-a-b)\right)}{\pi}\;\mathbf{F}\left({a,b\atop c};z\right)
&= \frac{1}{\Gamma\left(c-a\right)\Gamma\left(c-b\right)}\; \mathbf{F}\left({a,b%
\atop a+b-c+1};1-z\right) \\[0.3em]
&\quad\: -\frac{(1-z)^{c-a-b}}{\Gamma\left(a\right)\Gamma\left%
(b\right)}\; \mathbf{F}\left({c-a,c-b\atop c-a-b+1};1-z\right) \:,
\end{align*}
the last argument of these hypergeometric functions can be transformed in such a way that
we can again work with the simple asymptotics~\eqref{zzero}.
A straightforward computation yields
\[ 
\begin{split} 
\begin{pmatrix} u_1^+ \\[0.3em] u_2^+ \end{pmatrix}
&= \frac{1}{\cosh(m \pi)}\;
\begin{pmatrix} \displaystyle e^{-i mT} \: \frac{\pi \:\Gamma(\frac{1}{2}+im)}{\Gamma(\frac{1}{2}-im)\: \Gamma(\frac{1}{2}+im-\lambda)\: \Gamma(\frac{1}{2}+im+\lambda)} \\[1.5em] -i e^{i mT}\: \sin(\pi \lambda)
\end{pmatrix} + \O \big( e^{T} \big) \\[.5em]
\begin{pmatrix} u_1^- \\[0.3em] u_2^- \end{pmatrix}
&= \frac{1}{\cosh(m \pi)}\;
\begin{pmatrix} -i e^{-i mT}\: \sin(\pi \lambda) \\[0.5em]
\displaystyle e^{i mT} \: \frac{\pi \:\Gamma(\frac{1}{2}-im)}{\Gamma(\frac{1}{2}+im)\: \Gamma(\frac{1}{2}-im-\lambda)\: \Gamma(\frac{1}{2}-im+\lambda)}
\end{pmatrix} + \O \big( e^{T} \big) \:.
\end{split} 
\]

\subsection{Computation of the Fermionic Signature Operator}

Using the information gathered up to this point, we have all necessary ingredients to compute explicitly the fermionic signature operator in full de Sitter spacetime. More precisely we shall apply Lemma \ref{lemmamode} and Corollary \ref{fermionic_signature}, which guarantee us that, since the error term in \eqref{eq:error_term} decays exponentially in time, the operators $\Sig_m^\pm$ can be read off from the asymptotic expansion for large values of $T$ of the solutions of the Dirac equation, constructed in Section \ref{Sec:asymp_sol}. We obtain
\[ 
\Sig_m^+=\left(\begin{array}{cc}
1 & 0 \\
0 & -1 
\end{array}\right)
\]
and
\begin{align*}
&\Sig_m^- \\
&\!=\left( \!\!\begin{array}{cc}
\frac{\cos(2\pi\lambda)+\sinh^2(\pi m)}{\cosh^2(\pi m)} & \frac{2\pi i\Gamma(\frac{1}{2}+im)\sin(\pi\lambda)}{\cosh^2(\pi m)\Gamma(\frac{1}{2}-im)\Gamma(\frac{1}{2}+im+\lambda)\Gamma(\frac{1}{2}+im-\lambda)} \\
-\frac{2\pi i\Gamma(\frac{1}{2}-im)\sin(\pi\lambda)}{\cosh^2(\pi m)\Gamma(\frac{1}{2}+im)\Gamma(\frac{1}{2}-im+\lambda)\Gamma(\frac{1}{2}-im-\lambda)} & -\frac{\cos(2\pi\lambda)+\sinh^2(\pi m)}{\cosh^2(\pi m)} 
\end{array} \!\!\right).
\end{align*}
Using the formulae in Corollary \ref{fermionic_signature}, we obtain for the fermionic signature operator
\beq \label{Sig} \begin{split}
&\Sig_m \\
&\!= \left( \!\!\begin{array}{cc}
\frac{\cos(2\pi\lambda)+\cosh(2 \pi m)}{2\cosh^2(\pi m)} & \frac{2\pi i\Gamma(\frac{1}{2}+im)\sin(\pi\lambda)}{\cosh^2(\pi m)\Gamma(\frac{1}{2}-im)\Gamma(\frac{1}{2}+im+\lambda)\Gamma(\frac{1}{2}+im-\lambda)} \\
-\frac{2\pi i\Gamma(\frac{1}{2}-im)\sin(\pi\lambda)}{\cosh^2(\pi m)\Gamma(\frac{1}{2}+im)\Gamma(\frac{1}{2}-im+\lambda)\Gamma(\frac{1}{2}-im-\lambda)} & -\frac{\cos(2\pi\lambda)+\cosh(2 \pi m)}{2\cosh^2(\pi m)} 
\end{array} \!\!\right).
\end{split}
\eeq

Following Theorem \ref{Thm:2-pt}, starting from the fermionic signature operator, we can construct a unique bi-distribution $\mathcal{P}\in\mathcal{D}^\prime(\scrM\times\scrM)$, such that, for all $\phi,\psi\in C^\infty_0(\scrM,S\scrM)$ 
\begin{equation}\label{2-pt-BD}
\langle\phi,P\psi\rangle=\mathcal{P}(\phi\otimes\psi),
\end{equation}
where $P:=-\chi_{(-\infty, 0)}(S_m)k_m:C^\infty_0(\scrM,S\scrM)\to\mathcal{H}_m$ is defined as in \eqref{Pdef} starting from \eqref{Sig}. In turn, on account of Theorem \ref{thmstate}, this bi-distribution identifies a pure quasi-free state for the algebra of smeared field such that
$$\omega(\Psi(\phi)\Psi^*(\psi))=-\mathcal{P}(\phi\otimes\psi).$$

In order to analyze the physical significance of the obtained state, we observe that it enjoys the following two properties:

\begin{enumerate}[leftmargin=2.5em]
	\item On account of Lemma \ref{Lem:isometry_invariance}, $\mathcal{P}$ is a bi-distribution invariant under the action of the isometry group of full de Sitter. In other words, the state built out of $P$ is maximally symmetric. 
	\item In view of Equation \eqref{antismeared}, the integral kernel associated with the canonical anti-commutation relations coincides with that of the causal fundamental solution of the massive Dirac equation, that it $\tilde{k}_m$.
\end{enumerate}

\noindent These two properties entail the following 

\begin{Prp}
	Let $\omega$ be the quasi-free pure state which is unambiguously defined by \eqref{2-pt-BD} according to Theorem \ref{thmstate}. Then it coincides with the spinorial Bunch-Davies state. 
\end{Prp}

\begin{proof}
	Since $\omega$ determines a quasi-free state, it suffices to consider the two-point distribution. In particular, the relation \eqref{antismeared} implies that, for all $\phi,\psi\in C^\infty_0(\scrM,S\scrM)$, $\omega(\Psi(\phi)\Psi^*(\psi)+\Psi^*(\psi)\Psi(\phi))=\langle\phi^*|\tilde{k}_m\psi\rangle$, where $\tilde{k}_m$ is the causal fundamental solution.
	
	In \cite{Allen} it is shown that the space of the maximally symmetric two-point distributions associated to the Dirac equation on full de Sitter spacetime can be reconstructed out of two linearly independent bi-distributions, solutions of an ordinary differential equation in the geodesic distance as the underlying variable. The first one is built out of \eqref{eq:aux1} and of \eqref{eq:sol1a} and \eqref{eq:sol1b}. The singular structure of the hypergeometric function in a neighborhood of $Z(\mu)=0$ entails that the ensuing two-point function is compatible both with the Hadamard form and with Equation \eqref{antismeared}. The second linearly dependent solution on the other hand is of the same form as \eqref{eq:aux1} and of \eqref{eq:sol1a} and \eqref{eq:sol1b}, but with $Z(\mu)$ (defined in \eqref{eq:Zmu}) replaced by $1-Z(\mu)$. In particular, this implies that any two-point function of a state for the algebra of smeared fields cannot be compatible with \eqref{antismeared}. Hence, possibly up to an irrelevant normalization constant, the state built out of \eqref{2-pt-BD} must coincide with the spinorial Bunch-Davies state.
\end{proof}

\Thanks {{\em{Acknowledgments:}} C.D. is grateful to the Department of Mathematics of the University of Regensburg and the Mathematical Institute of the University of Freiburg for the kind hospitality during the realization of part of this work. S.M.\ acknowledges the  support by the research grant { ``Geometric boundary value problems for the Dirac operator''} of the Juniorprofessurenprogramm Baden-W\"urttemberg as well as that by the DFG research training groups GRK 1692 ``Curvature, Cycles, and Cohomology'' and  GRK 1821 ``Cohomological Methods in Geometry'' during the initial stages of this project. 
	E.R. wishes to thank the Department of Mathematics of the University of Regensburg for the kind hospitality during the realization of part of this project and she acknowledges the support of the INdAM project GNAMPA2017  ``Analisi Di Modelli Matematici
Della Fisica, Della Biologia E Delle Scienze Sociali" for the initial stages of this work.


\end{document}